\documentclass[11pt,finalcls,onecolumn]{IEEEtran}

\usepackage{Ramble}

\usepackage{color}
\usepackage{soul}
\usepackage{romannum}

\begin{document}
\pagenumbering{arabic}
\title{Design of Bilayer and Multi-layer LDPC Ensembles from Individual Degree Distributions}
\author{
\vspace*{0.8cm} Eshed Ram \qquad Yuval Cassuto\\
Andrew and Erna Viterbi Department of Electrical Engineering \\
Technion -- Israel Institute of Technology, Haifa 32000, Israel\\
E-mails: \{s6eshedr@campus, ycassuto@ee\}.technion.ac.il
}

\maketitle

\begin{abstract}
	A new approach for designing bilayer and multi layer LDPC codes is proposed and studied in the asymptotic regime. The ensembles are defined through individual uni-variate degree distributions, one for each layer. We present a construction that: 1) enables low-complexity decoding for high-SNR channel instances, 2) provably approaches capacity for low-SNR instances, 3) scales linearly (in terms of design complexity) in the number of layers. For the setup where decoding the second layer is significantly more costly than the first layer, we propose an optimal-cost decoding schedule and study the trade-off between code rate and decoding cost.
	\footnote{Part of the results of this paper were presented at the 2018 International Symposium on Information Theory.}
\end{abstract}

\section{Introduction}
\label{Sec:Intro}

Low-Density Parity-Check (LDPC) codes and their message-passing decoding algorithms \cite{Gala62} are an efficient way to achieve Shannon's limit on various channels. LDPC codes are extremely powerful because they can attain competitive performance with low-complexity decoding (message passing on sparse bi-partite graphs) and simple code design (random drawing from explicit code ensembles). Many extensions of the basic LDPC construction have been proposed to enhance the code's functionality through more structured code graphs, while maintaining the convenience of randomly drawing codes from ensembles. Some prominent examples are repeat-accumulate codes \cite{DivsalarJin98,JinMcEliece00}, protograph-based codes \cite{Thorpe03}, spatially-coupled codes \cite{KudRichUrb11}, and multi-edge-type (MET) codes \cite{RichUrb08} (which can be viewed as a meta-class containing all the others). An especially useful class of structured LDPC codes is {\em bilayer codes} \cite{RazaYu07,AzmiYuan08,RazYu06,WeinMart09} (and more generally multi-layer codes), which allow decoding the same codeword with two (or more) different decoders. A recent interest in bilayer LDPC codes is raised by storage applications, in which the code needs to be designed for both extreme and average channel conditions \cite{RamCassuto18a,ShaAl20}.

In this paper we develop new tools for the construction and (asymptotic) analysis of bilayer LDPC code ensembles. An important feature of these tools is that they lend themselves well for extending beyond two layers, which we do later in the paper. A bilayer LDPC code has a code graph in which the variable nodes are connected to two types of check nodes, allowing different connectivity to each type. There are several applications that motivate codes with bilayer structure: communication over the relay channel \cite{RazaYu07}, decoder parallelization \cite{Li17}, multi-block coding \cite{RamCassuto18a}, incremental redundancy \cite{WangRangWesel17,WangRangWesel18,WangRangWesel19}, to name a few. Bilayer codes are a special case of MET codes that are defined at full generality in \cite{RichUrb08}. However, to gain tools and insight for construction and analysis, we find it beneficial to use a more compact specification of the codes than needed when viewed as MET codes. In particular, this compact specification is the key to our ability to extend the results to more than two layers without complexity blow-up.

\subsection{Contributions}
\label{Sub:contributions}
Bilayer LDPC codes are designed to simultaneously guarantee correction capabilities in two decoding modes: layer 1 only, and layers 1+2. More generally, $L$-layer codes are designed for $L$ decoding modes: layer 1, layers 1+2, ... , layers 1+2+$\cdots$+$L$. The design of such codes in this paper is pursued through a new approach: {\em the codes are defined by specifying each layer separately as a standard degree-distribution pair}. This approach may at first seem less natural given that we decode layers 1+2 jointly, and not layer 2 separately. Indeed, prior work \cite{RazaYu07} designed the joint layer-1+2 code directly. The advantage of the new approach is that working with standard degree distributions (specified as uni-variate polynomials), rather than product degree distributions (specified as multi-variate polynomials), enables tractable design of explicit code ensembles with provable asymptotic performance. 

The basis of this approach is laid in Section~\ref{Sec:2D DE}, where the correction capability of the layer-1+2 code is characterized mathematically given the separate layer-1 and layer-2 degree distributions. This is achieved by deriving a two-dimensional density-evolution framework, where decoding thresholds are found as certain fixed points in two variables (each variable tracks the density on edges of one layer). 

Section~\ref{Sec:C1} provides a general construction for bilayer codes with any desired thresholds for layer-1 and layer-1+2 decoding. The resulting codes are given as explicit degree distributions (building on known properties of standard single-layer codes), without need to employ optimization tools such as linear programming. In particular, the construction is used to construct code sequences that approach capacity for layer-1+2 decoding, while guaranteeing any desired threshold for layer-1 decoding. Furthermore, the additive gap to capacity of the layer-1+2 code is characterized (and bounded) given the gaps to capacity of the individual layer-1 and layer-2 degree distributions. Section~\ref{Sec:Multi} generalizes the results of Sections~\ref{Sec:2D DE},\ref{Sec:C1} to $L$-layer codes, for any $L\geq 2$. 

In Section~\ref{Sec:N_JI} we treat a model in which layer-2 decoding iterations are more costly than layer-1, and thus we seek codes that successfully decode layers 1+2 with few layer-2 iterations. For this model we propose optimal-cost decoding schedules, and study the trade-off between rate (through layer-1 and layer-2 gaps to capacity) and decoding cost. 

All of our results are given for coding over the binary erasure channel (BEC), however, the same analysis and constructions extend to other channels (such as the AWGN channel) using the EXIT method.

\subsection{Related Work}
\label{Sub:relate}
Bilayer LDPC codes are used to implement binning for the relay channel in \cite{RazaYu07,AzmiYuan08,RazYu06} (\cite{WeinMart09} also suggests bilayer codes with one layer being a low-density generator-matrix (LDGM) code). 
The design approach in these works is to first optimize layer 1 for the source-relay channel (higher SNR), and then optimize layers 1+2 for the source-destination channel (lower SNR) constrained to be consistent with layer 1. This approach extends the classical single-layer linear-programming ensemble-design framework by taking the layer 1+2 optimization variables to be coefficients of a {\em bi-variate} degree-distribution polynomial, and adding the layer-1 consistency constraints. This extension, however, entails solving an optimization problem with many variables: the {\em product} of the maximal degrees in layer 1 and layer 2. Furthermore, extending this approach to $L$-layer codes would make the number of variables grow exponentially with $L$ (viewed as MET codes \cite{RichUrb08}, the canonical representation of an $L$-layer ensemble is by an $L$-variate polynomial, whose number of coefficients grows exponentially with $L$ for a given maximum layer degree.) Our results in this paper provide a design alternative that avoids this multiplicative/exponential growth of complexity, and also extend beyond optimization frameworks to offer analytical insight. 

Relevant to our techniques is the work in \cite{BarEr10} where (single-layer) LDPC codes are used for a BEC with an erasure rate that can take one of two values: full decoding is sought for the better channel, while for the worse channel some partial correction performance is specified. More related work includes non-asymptotic design of bilayer LDPC codes such as codes that enable decoder parallelism \cite{Li17} and codes for incremental redundancy \cite{WangRangWesel17,WangRangWesel18,WangRangWesel19} (whose layer 2 is an LDGM code). Another indication of interest in bilayer/multi-layer codes is a parallel work on (non-LDPC) algebraic bilayer/multi-layer codes  \cite{Hass01, HanMont07, BlauHetz16, CassHemo17}.

\section{Preliminaries and Notations}
\label{Sec:Pre}

\subsection{LDPC Codes}
\label{sub:LDPC}

A linear block code is an LDPC code if it has at least one parity-check matrix that is sparse, i.e., the number of 1's in $H$ is linear in the block length. 
Every parity-check matrix $H$ can be represented by a bipartite graph, called a Tanner graph, with nodes partitioned to variable nodes and check nodes; there exists an edge between check node $i$ and variable node $j$, if and only if $H_{ij}=1$ (this paper focuses on binary linear codes, but this representation can be generalized). In single-edge-type LDPC codes, the fraction of variable (resp. check) nodes in a Tanner graph with degree $i$ is denoted by $\Lambda_i$ (resp. $\Omega_i$), and the fraction of edges connected to variable (resp. check) nodes of degree $i$ is denoted by $\lambda_i$ (resp. $\rho_i$); $\Lambda_i$ and $\Omega_i$ are called node-perspective degree distributions, and $\lambda_i$ and $\rho_i$ are called edge-perspective degree distributions. 
The degree-distribution polynomials associated to a Tanner graph are given by
\begin{align} 
\label{eq:dd poly1}
&\Lambda(x)=\sum_{i}\Lambda_{i} x^i,    \qquad  \lambda(x)=\sum_{i}\lambda_{i} x^{i-1},\quad x \in [0,1],\\
\label{eq:dd poly2}
& \Omega(x)=\sum_{i}\Omega_{i}x^i ,\qquad \rho(x) = \sum_{i}\rho_{i} x^{i-1},\quad x \in [0,1].
\end{align}

\subsection{Bilayer LDPC Codes}
\label{sub:Tanner}
\subsubsection{Graph Structure}
In the bilayer Tanner graph, the check nodes are divided into two disjoint sets: \emph{type-1} and \emph{type-2} check nodes. 
The edges of the graph are partitioned into two sets as well: edges connecting variable nodes to type-1 (resp. type-2) check nodes are type-1 (resp. type-2) edges (see Figure~\ref{Fig:MB Tanner}). Finally, each variable node has two types of degrees, corresponding to the two edge types.  

Bilayer codes can be generalized to allow multi-block codes \cite{RamCassuto18a} where the sets of variable nodes and type-1 check nodes are partitioned into disjoint subsets, each being connected ``locally", and type-2 check nodes are connected across all subsets.
Our results in the sequel carry over to this generalization, but for simplicity we present them for the simpler structure depicted in Figure~\ref{Fig:MB Tanner}.
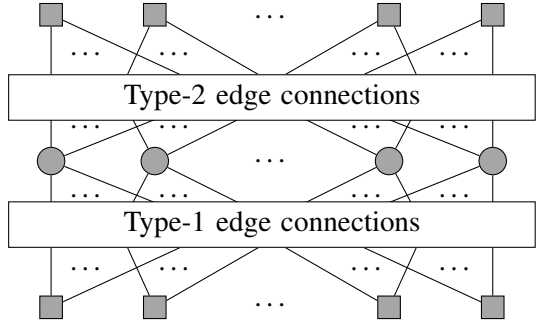
\begin{figure}[!h]

	\begin{center}
     	\begin{tikzpicture}[>=latex]\label{Tikz:bilayer tanner}
     	\tikzstyle{cnode}=[rectangle,draw,fill=gray!70!white,minimum size=3mm]
     	\tikzstyle{vnode}=[circle,draw,fill=gray!70!white,minimum size=2mm]
		\pgfmathsetmacro{\x}{10}
		\pgfmathsetmacro{\w}{1.5}
		\pgfmathsetmacro{\y}{4}
     	\node (v1) [vnode] at (0,0) {}; \node (v2) [vnode,right=\x mm of v1] {};
     	\node (dotsV) [right=\x mm of v2] {\dots};\node (v3) [vnode,right=\x mm of dotsV] {};
     	\node (v4) [vnode,right=\x mm of v3] {};
     	\node (mid) at ($(v3)!0.5!(v2)$) {};
     	\node (midu) [above = 1.6*\y mm of mid] {};
     
     	\node (midd) [below = 1.6*\y mm of mid] {};
     	\node (c11) [cnode,below = 4*\y mm of v1] {}; \node (c12) [cnode,below = 4*\y mm of v2] {}; 
     	\node (dotsC1) [below = 4*\y mm of dotsV] {\dots};  \node (c13) [cnode,below = 4*\y mm of v3] {}; 
     	\node (c14) [cnode,below = 4*\y mm of v4] {}; 
     	\node (c21) [cnode,above = 4*\y mm of v1] {}; \node (c22) [cnode,above = 4*\y mm of v2] {}; 
     	\node (dotsC2) [above = 4*\y mm of dotsV] {\dots};  \node (c23) [cnode,above = 4*\y mm of v3] {}; 
     	\node (c24) [cnode,above = 4*\y mm of v4] {}; 
     	
     	\node  (l2) at (c21|-midu) {};	\node  (r2) at (midu-|c24) {};	
     	\node  (l1) at (c11|-midd) {};	\node  (r1) at (midd-|c14) {};	
     	\node (l21) at ($(l2)!0.33!(midu)$) {};\node (l22) at ($(l2)!0.66!(midu)$) {};
     	\node (r21) at ($(r2)!0.33!(midu)$) {};\node (r22) at ($(r2)!0.66!(midu)$) {};
     	\node (l11) at ($(l1)!0.33!(midd)$) {};\node (l12) at ($(l1)!0.66!(midd)$) {};
     	\node (r11) at ($(r1)!0.33!(midd)$) {};\node (r12) at ($(r1)!0.66!(midd)$) {};
     	
     	\draw (c21)--(l2.north);\node at($(c21)!0.5!(l21)$) {\dots} ; \draw (c21)--(l22.north);
     	\draw (c22)--(l21.north);\node at($(c22)!0.5!(l22)$) {\dots} ; \draw (c22)--(midu.north);
     	\draw (c24)--(r2.north);\node at($(c24)!0.5!(r21)$) {\dots} ; \draw (c24)--(r22.north);
     	\draw (c23)--(r21.north);\node at($(c23)!0.5!(r22)$) {\dots} ; \draw (c23)--(midu.north);
     	
     	\draw (c11)--(l1.south);\node at($(c11)!0.5!(l11)$) {\dots} ; \draw (c11)--(l12.south);
     	\draw (c12)--(l11.south);\node at($(c12)!0.5!(l12)$) {\dots} ; \draw (c12)--(midd.south);
     	\draw (c14)--(r1.south);\node at($(c14)!0.5!(r11)$) {\dots} ; \draw (c14)--(r12.south);
     	\draw (c13)--(r11.south);\node at($(c13)!0.5!(r12)$) {\dots} ; \draw (c13)--(midd.south);
     	
     	\draw (v1)--(l1.north);\node at($(v1)!0.5!(l11)$) {\dots} ; \draw (v1)--(l12.north);
       	\draw (v2)--(l11.north);\node at($(v2)!0.5!(l12)$) {\dots} ; \draw (v2)--(midd.north);
       	\draw (v4)--(r1.north);\node at($(v4)!0.5!(r11)$) {\dots} ; \draw (v4)--(r12.north);
       	\draw (v3)--(r11.north);\node at($(v3)!0.5!(r12)$) {\dots} ; \draw (v3)--(midd.north);
       	
       	\draw (v1)--(l2.south);\node at($(v1)!0.5!(l21)$) {\dots} ; \draw (v1)--(l22.south);
       	\draw (v2)--(l21.south);\node at($(v2)!0.5!(l22)$) {\dots} ; \draw (v2)--(midu.south);
       	\draw (v4)--(r2.south);\node at($(v4)!0.5!(r21)$) {\dots} ; \draw (v4)--(r22.south);
       	\draw (v3)--(r21.south);\node at($(v3)!0.5!(r22)$) {\dots} ; \draw (v3)--(midu.south);
       	
     	\node (midu) [draw,above = \y mm of mid,minimum width=7*\x mm,fill=white] {Type-2 edge connections};
    	\node (midd) [draw,below = \y mm of mid,minimum width=7*\x mm,fill=white] {Type-1 edge connections};

     	\end{tikzpicture}
     \end{center}

     \caption{\label{Fig:MB Tanner}
      An illustration of a bilayer Tanner graph}

\end{figure}

We denote by $\Lambda^{(1)}_{i}$ (resp. $\Lambda^{(2)}_{i}$) the fraction of variable nodes with type-1 (resp. type-2) degree $i$. Similarly, $\Omega^{(1)}_{i}$ (resp. $\Omega^{(2)}_{i}$) is the fraction of type-1 (resp. type-2)  check nodes of degree $i$. $\lambda^{(1)}_{i}$ (resp. $\lambda^{(2)}_{i}$) designates the fraction of type-1 (resp. type-2) edges connected to a variable node with type-1 (resp. type-2) degree $i$, and $\rho^{(1)}_{i}$ (resp. $\rho^{(2)}_{i}$) designates the fraction of type-1 (resp. type-2) edges connected to a type-1 (resp. type-2) check node of degree $i$. Note that since the type-1 sub-graph (i.e., variable nodes, type-1 check nodes, and type-1 edges) is supposed to be used in decoding without the type-2 check nodes and edges, then $\Lambda^{(1)}_{0}=\Lambda^{(1)}_{1}=0$. However, this is not the case with the type-2 sub-graph (which is assumed to be used together with the type-1 sub-graph), and we do allow $\Lambda^{(2)}_{0},\Lambda^{(2)}_{1}>0$ (see also \cite{RazaYu07}). In sections~\ref{Sec:2D DE}--\ref{Sec:C1}, we will use $P_0$ to denote the coefficient $\Lambda^{(2)}_{0}$.
The type-1 and type-2 degree-distribution polynomials $\Lambda^{(i)}(\cdot),\lambda^{(i)}(\cdot),\Omega^{(i)}(\cdot),\rho^{(i)}(\cdot),\, i\in\{1,2\}$ are defined similarly to the degree-distribution polynomials for ordinary LDPC codes in \eqref{eq:dd poly1}-\eqref{eq:dd poly2}. 

The ensembles induced by the above description of bilayer graphs are characterized by the block length $n$, and the above degree distributions. The design rate of this ensemble is given by
\begin{align}\label{eq:design rate2}
R =  1 - \frac{\int_0^1\rho^{(1)}(x)\mathrm{d}x}{\int_0^1\lambda^{(1)}(x)\mathrm{d}x} -  \frac{\int_0^1\rho^{(2)}(x)\mathrm{d}x}{\int_0^1\lambda^{(2)}(x)\mathrm{d}x}\left( 1-P_0\right).
\end{align}
We can see in \eqref{eq:design rate2}, that setting $P_0>0$ allows increasing the code rate, which we later find crucial in our constructions.

\subsubsection{Density Evolution}
We distinguish between decoding using only layer 1, i.e., the type-1 sub-graph, and decoding using the entire graph (both layers). This separation yields a two-level threshold behavior: when using only layer 1, the asymptotic threshold is denoted by $\epsilon_1^*$, and when using both layers the threshold is marked as $\epsilon_2^*$. Since the second layer can only help, then $\epsilon_2^*\geq \epsilon_1^*$. However, decoding using layer 1 only has complexity advantages, and if the signal-to-noise ratio (SNR) is high enough, then layer 1 suffices. Moreover, in some applications layer-1 decoding is performed on fewer variable nodes comprising a sub-block of the full code block (see \cite{RamCassuto18a}), further reducing the complexity.

When decoding the type-1 sub-graph, one can use the known density evolution method for LDPC codes to analyze the performance. Specifically, 
the fraction of variable-to-check erasure messages after $l$ BP iterations over the BEC($\epsilon$), $x_l(\epsilon)$, is given by the recursive equation \cite{RichUrb08}
\begin{align*}
x_l(\epsilon) = \epsilon \lambda^{(1)}(1-\rho^{(1)}(1-x_{l-1})).
\end{align*}
From this, $\epsilon^{*}_1$ can be calculated via 
\begin{align} \label{eq:1D th numeric}
\epsilon^{*}_1 = \inf_{x\in(0,1]} \frac{x}{\lambda^{(1)}(1-\rho^{(1)}(1-x))}.
\end{align}

When decoding both layers, one should consider both degree-distribution pairs for the analysis, since, as a specific instance of MET codes, the graph structure plays a crucial rule in the decoding analysis. 
While MET codes can be specified in full generality using degree-distributions multinomials \cite{RichUrb08}, more compact representations are typically helpful for obtaining effective analysis and design tools for particular classes of MET codes.
For example, in \cite{RazaYu07}, it is shown that for a variable node with type-1 and type-2 degrees $ i $ and $ j $, respectively, and incoming type-1 and type-2 densities (erasure rates) $ u \in[0,1]$ and $w \in[0,1]$, respectively, the outgoing type-1 and type-2 densities $ x $ and $y$, respectively, are given by \cite[Eq. (15)--(16)]{RazaYu07}
\begin{align*}
x=\epsilon\cdot w^{j} \cdot u^{i-1},\quad y=\epsilon \cdot w^{j-1} \cdot u^{i}.
\end{align*}
The authors in \cite{RazaYu07} define the bilayer code through a product variable-node degree distribution specified by a bivariate polynomial with coefficients $ \lambda_{i,j} $ ($ i $ for layer 1 and $ j $ for layer 2) and regular check node degrees. They then pursue code design using linear-programming optimization of the product degree distribution  $ \lambda_{i,j} $, under the constraint that it is consistent with a given (capacity approaching) degree distribution for layer 1.

We take a different approach and specify (and then design) the code ensembles through the individual degree distributions $\left (\lambda^{(1)}(\cdot),\rho^{(1)}(\cdot)\right ) $ and $\left (\lambda^{(2)}(\cdot),\rho^{(2)}(\cdot),P_0\right )$. As we will later see, this approach offers analysis and construction advantages compared to \cite{RazaYu07}.

\section{Threshold}\label{Sec:2D DE}
In this section, we study the asymptotic decoding threshold of bilayer ensembles, as defined in Section~\ref{Sec:Pre}, when using both layers. 
Our ultimate goal (in Section~\ref{Sec:C1}) is to provide a design tool for constructing bilayer LDPC codes: given two noise levels, $ \epsilon_1 $ and $ \epsilon_2 $, choose $\left (\lambda^{(1)}(\cdot),\rho^{(1)}(\cdot)\right ) $ such that layer 1 provides the correction capability to tolerate $ \epsilon_1 $, and then set $\left (\lambda^{(2)}(\cdot),\rho^{(2)}(\cdot),P_0\right )$ such that the entire graph provides the correction capability to tolerate $ \epsilon_2 $. Working with $\left (\lambda^{(2)}(\cdot),\rho^{(2)}(\cdot),P_0\right )$ instead of the product degree distribution as in \cite{RazaYu07} enables, among other benefits, finding capacity-approaching sequences for the full-graph  code.
The derivations in this section lay the theoretical infrastructure needed to show the optimality of our constructions (i.e., capacity achieving in Section~\ref{Sec:C1}). 

We assume that the message scheduling when decoding the entire graph is a flooding schedule: the variable nodes send messages to the type-1 and type-2 check nodes in parallel (later in Section~\ref{Sec:N_JI}, we change the schedule from flooding to be more complexity aware). 
Consider a random instance from the bilayer ensemble characterized by block length $ n $ and degree distributions $\Lambda^{(1)},\Lambda^{(2)},\Omega^{(1)},\Omega^{(2)}$ (recall that $\Lambda^{(2)}$ includes $P_0$). Let $x_l(\epsilon)$ and $y_l(\epsilon)$ denote the probability that a type-1 and type-2 edge, respectively, carry a variable-to-check erasure message after $l$ BP iterations over the BEC($\epsilon$) as $n \to \infty$. In view of the MET density-evolution equations in \cite{RichUrb08}, we have
\begin{subequations}
	\begin{align}
	\label{eq:2D DE1}
	&x_l(\epsilon) = \epsilon \cdot \lambda^{(1)}\left(1- \rho^{(1)}\left(1-x_{l-1}(\epsilon) \right)\right)\cdot \Lambda^{(2)}\left(1- \rho^{(2)}\left(1-y_{l-1}(\epsilon) \right)\right), \quad l \geq 0,\\
	\label{eq:2D DE2}
	&y_l(\epsilon) = \epsilon \cdot \Lambda^{(1)}\left(1- \rho^{(1)}\left(1-x_{l-1}(\epsilon) \right)\right) \cdot \lambda^{(2)}\left(1- \rho^{(2)}\left(1-y_{l-1} (\epsilon) \right)\right),\quad l \geq 0,\\
	\label{eq:2D DE3}
	&x_{-1}(\epsilon)=y_{-1}(\epsilon)=1.
	\end{align}
\end{subequations}
Figure~\ref{Fig:DE Eq} graphically illustrates equations \eqref{eq:2D DE1}--\eqref{eq:2D DE2}: in the center diagram the right outgoing edge carries the message in \eqref{eq:2D DE1} to a type-1 check and the left outgoing edge carries the message in \eqref{eq:2D DE2} to a type-2 check. 
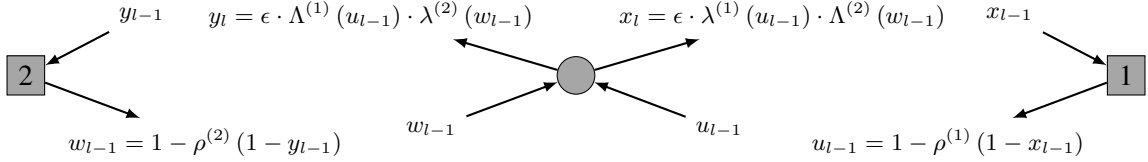
\begin{figure}[!h]
	
	\begin{center}
		\begin{tikzpicture}[scale=0.3,>=latex]\label{Tikz:DE eq}
		\tikzstyle{cnode}=[rectangle,draw,fill=gray!70!white,minimum size=5mm]
		\tikzstyle{vnode}=[circle,draw,fill=gray!70!white,minimum size=5mm]
		\pgfmathsetmacro{\x}{17}
		\pgfmathsetmacro{\w}{4*\x}
		\pgfmathsetmacro{\y}{4}
		
		\node[cnode] (cJ) {2};
		\node [above right = 3mm and 0.5*\x mm of cJ] (incJ) {\footnotesize $ y_{l-1} $};
		\node [below right = 3mm and 0.1*\x mm of cJ] (outcJ) {\footnotesize $ w_{l-1}=1-\rho^{(2)}\left (1-y_{l-1}\right ) $};
		\draw[->,thick] (incJ)--(cJ) ;
		\draw[->,thick] (cJ)--(outcJ) ;
		
		\node[vnode,right=\w mm of cJ] (v) {};
		\node [above right = 3mm and 0.15*\x mm of v]  (outv2) {\footnotesize $ x_{l}=\epsilon\cdot\lambda^{(1)}\left (u_{l-1} \right )\cdot\Lambda^{(2)}\left (w_{l-1} \right)$};
		\node [below left = 3mm and 0.75*\x mm of v] (invJ) {\footnotesize $ w_{l-1} $};
		\node [above left = 3mm and 0.15*\x mm of v] (outv1) {\footnotesize $ y_{l}=\epsilon\cdot\Lambda^{(1)}\left (u_{l-1} \right )\cdot\lambda^{(2)}\left (w_{l-1} \right)$};
		\node [below right = 3mm and 0.75*\x mm of v] (invL) {\footnotesize $ u_{l-1} $};;
		\draw[->,thick] (invL)--(v) ;
		\draw[->,thick] (invJ)--(v) ;
		\draw[->,thick] (v)--(outv1) ;
		\draw[->,thick] (v)--(outv2) ;
		
		\node[cnode,right=\w mm of v] (cL) {1};
		\node [above left = 3mm and 0.5*\x mm of cL] (incL) {\footnotesize $ x_{l-1} $};
		\node [below left = 3mm and 0.1*\x mm of cL] (outcL) {\footnotesize $ u_{l-1}=1-\rho^{(1)}\left (1-x_{l-1}\right ) $};
		\draw[->,thick] (incL)--(cL);
		\draw[->,thick] (cL)--(outcL) ;
		
		\end{tikzpicture}
	\end{center}
	
	\caption{\label{Fig:DE Eq} Illustration of the DE equations \eqref{eq:2D DE1}--\eqref{eq:2D DE2}.}
	
\end{figure}
To simplify notations, $\epsilon$ will be omitted from now on if it is clear from the context. 
\begin{remark} \label{remark:asymmetry}
	Although $x_l$ and $y_l$ in \eqref{eq:2D DE1}-\eqref{eq:2D DE2} seem symmetric to each other, it is not necessarily true since we allow variable nodes to have type-2 degrees $0$ ($P_0>0$) or $1$ ($\lambda^{(2)}(0)>0$), while their type-1 degrees are forced to be strictly greater then $1$. This asymmetry has a crucial effect on the decoding process which is explained and detailed later in this section.
	Symmetry does hold in the special and less interesting case where $\rho^{(1)}(x)=\rho^{(2)}(x)$ and $\lambda^{(1)}(x)=x^{l_1-1},\lambda^{(2)}(x)=x^{l_2-1}$, in which case, for every iteration $ l\geq 0 $, $y_l =x_l=\epsilon \lambda(1-\rho(1-x_{l-1})),$ where, $\rho(x)\triangleq\rho^{(1)}(x)$, $\lambda(x)\triangleq x^{l_1+l_2-1}$. Thus, if we use identical degree-distributions for type-1 and type-2 check nodes and we force all variable nodes to have regular degrees, then the 2D-DE equations in \eqref{eq:2D DE1}-\eqref{eq:2D DE2} degenerate to the already known 1D-DE equation. However, codes falling under this special case are less interesting because they are sub-optimal in their rates and restricted in their thresholds. 
\end{remark}
Define 
\begin{subequations}
\begin{align} 
\label{eq:f}
f(\epsilon,x,y)=\epsilon\cdot\lambda^{(1)}\left(1- \rho^{(1)}\left(1-x \right)\right) \Lambda^{(2)}\left(1- \rho^{(2)}\left(1-y\right)\right),\quad x,y,\epsilon \in [0,1] \\
\label{eq:g}
g(\epsilon,x,y)=\epsilon\cdot\Lambda^{(1)}\left(1- \rho^{(1)}\left(1-x \right)\right)  \lambda^{(2)}\left(1- \rho^{(2)}\left(1-y\right)\right) ,\quad x,y,\epsilon \in [0,1]
\end{align}
\end{subequations}
such that \eqref{eq:2D DE1}-\eqref{eq:2D DE3} can be re-written as
\begin{subequations}
\begin{align} 
\label{eq:2D-DE1}
&x_l = f\left(\epsilon,x_{l-1},y_{l-1}\right) ,\quad l \geq 0\\
\label{eq:2D-DE2}
&y_l = g\left(\epsilon,x_{l-1},y_{l-1}\right) ,\quad l \geq 0\\
\label{eq:2D-DE3}
&x_{-1} =y_{-1} =1.
\end{align}
\end{subequations}

\begin{lemma} \label{lemma:monotonicity}
The functions $f$ and $g$ are monotonically non-decreasing in all of their variables.
\end{lemma}
\begin{proof}
Since the images of $\lambda^{(1)}(\cdot),\Lambda^{(1)}(\cdot),\lambda^{(2)}(\cdot),\Lambda^{(2)}(\cdot),\rho^{(1)}(\cdot)$ and $\rho^{(1)}(\cdot)$ lie in $[0,1]$, then $f$ and $g$ are monotonically non-decreasing in $\epsilon \in [0,1]$. The proof for $ x,y $ is similar.
\end{proof}

\begin{definition}
Let $\epsilon \in (0,1)$. We say that $(x,y) \in [0,1]^2$ is an $(f,g)$-fixed point if
\begin{align} \label{eq:f,g fixed}
\begin{pmatrix}
x\\y
\end{pmatrix}
=
\begin{pmatrix}
f(\epsilon,x,y)\\
g(\epsilon,x,y)
\end{pmatrix}.
\end{align}
\end{definition}

Clearly, for every $\epsilon \in (0,1)$, $(x,y)=(0,0)$ is a trivial $(f,g)$-fixed point. However, it is not clear yet if there exists a non-trivial $(f,g)$-fixed point. In particular, we ask: for which choices of $\epsilon,\; \lambda^{(1)},\rho^{(1)},\lambda^{(2)},\rho^{(2)}$ and $P_0$ there exists a non-trivial $(f,g)$-fixed point? The following lemmas help answering this question.

\begin{lemma} \label{lemma:fix pt}
Let $\epsilon \in (0,1)$, and let $(x,y) \in [0,1]^2$ be an $(f,g)$-fixed point. Then,
\begin{enumerate}
\item \label{item: fix pt lemma0} $x=0$ implies $y =0$, and if $P_0=0$ or $\lambda^{(2)}(0)>0$, then $y=0$ implies $x =0$.
\item \label{item: fix pt lemma1} $(x,y) \in [0,\epsilon)^2$.
\item \label{item: fix pt lemma2} If $\{x_{l}\}_{l=0}^\infty$ and $\{y_{l}\}_{l=0}^\infty$ are defined by \eqref{eq:2D-DE1}--\eqref{eq:2D-DE3}, then
\begin{align} \label{eq:no jumps}
x_{l} \geq x , \quad y_{l} \geq y,\quad \forall l\geq 0.
\end{align}
\end{enumerate}
\end{lemma}

\begin{proof}
See Appendix~\ref{App:fix pt}.
\end{proof}

\begin{remark} \label{remark: noting asymmetry}
Item 1 in Lemma~\ref{lemma:fix pt} expresses the asymmetry (discussed in Remark~\ref{remark:asymmetry}) between the type-1 and type-2 densities. Note that the ensembles where symmetry does \emph{not} hold are those that have $\Lambda^{(2)}_0>0$ and $\Lambda^{(2)}_1=0$, a common combination in our constructions later in the paper.
\end{remark}

\begin{lemma} \label{lemma:mono of DE}
Let $x_{l} $ and $y_{l}$ be defined by \eqref{eq:2D-DE1}--\eqref{eq:2D-DE3} and let $0<\epsilon \leq \epsilon' < 1$. Then,
\begin{subequations}
\begin{align} \label{eq:mono of DE1}
\begin{split}
x_{l+1}(\epsilon) \leq x_{l}(\epsilon) ,\quad y_{l+1}(\epsilon) \leq y_{l}(\epsilon) , \quad \forall l \geq 0,
\end{split}
\end{align}
and
\begin{align} \label{eq:mono of DE2}
\begin{split}
x_{l}(\epsilon) \leq x_{l}(\epsilon') , \quad y_{l}(\epsilon) \leq y_{l}(\epsilon') , \quad \forall l \geq 0.
\end{split}
\end{align}
\end{subequations}
\end{lemma}
\begin{proof}
By mathematical induction on $l$ and by Lemma~\ref{lemma:monotonicity}.
\end{proof}

In view of \eqref{eq:2D-DE1}--\eqref{eq:2D-DE3}, it can be verified that for every iteration $ l\geq 0 $, $x_l(0) = y_l(0) = 0 ,\; x_l(1) = y_l(1) = 1$. Since $x_l$ and $y_l$ are bounded from below by $0$, then  Lemma~\ref{lemma:mono of DE} implies that the limits $\lim_{l \to \infty}x_l(\epsilon)$ and $\lim_{l \to \infty}y_l(\epsilon)$ exist. Thus we can define a decoding threshold by
\begin{align} \label{eq:th op def A}
\epsilon^*_2 = \sup \left\{\epsilon \in [0,1] \colon \lim\limits_{l \to \infty}y_l(\epsilon)=\lim\limits_{l \to \infty}x_l(\epsilon)=0 \right\}.
\end{align}
Note that from the continuity of $g$ in \eqref{eq:g}, item~\ref{item: fix pt lemma0} in Lemma~\ref{lemma:fix pt} implies that if $\lim_{l \to \infty}x_l(\epsilon)=0 $, then  $\lim_{l \to \infty}y_l(\epsilon)=0 $. Thus, \eqref{eq:th op def A} can be re-written as
\begin{align} \label{eq:th op def}
\epsilon^*_2 = \sup \left\{\epsilon \in [0,1] \colon \lim\limits_{l \to \infty}x_l(\epsilon)=0 \right\}.
\end{align}

\begin{theorem}
\label{theorem:fix pt char}
$\epsilon^*_2 = \sup \left\{ \epsilon \in [0,1]\colon \eqref{eq:f,g fixed}\text{ has no solution with }(x,y) \in (0,1]\times[0,1]\right\}$.
\end{theorem}

\begin{proof}
Mark $\hat{\epsilon}=\sup \left\{ \epsilon \in [0,1]\colon \eqref{eq:f,g fixed}\text{ has no solution with }(x,y) \in (0,1]\times[0,1]\right\}$, let $\epsilon<\hat{\epsilon}$, and let $x(\epsilon)=\lim_{l \to \infty}x_l(\epsilon),\;y(\epsilon)=\lim_{l \to \infty}y_l(\epsilon)$. Taking the limit $l \to \infty$ in \eqref{eq:2D-DE1}--\eqref{eq:2D-DE3} yields that $(x(\epsilon),y(\epsilon))$ is an $(f,g)$-fixed point. Since $\epsilon<\hat{\epsilon}$, it follows that $x(\epsilon)=0$. From \eqref{eq:th op def} we have  $\epsilon < \epsilon^*_2$, for every $\epsilon<\hat{\epsilon}$; this implies that $\hat{\epsilon} \leq \epsilon^*_2$. 

For the other direction, let $\epsilon>\hat{\epsilon}$ and let $(z_1,z_2)$ be an $(f,g)$-fixed point such that $z_1>0$. Lemma~\ref{lemma:fix pt}-item~\ref{item: fix pt lemma2} implies that for every iteration $ l $, $x_l(\epsilon) \geq z_1 >0,\quad \forall l \geq 0$, thus $\lim_{l \to \infty}x_l(\epsilon) >0$, where the existence of this limit is assured due to Lemma~\ref{lemma:mono of DE}; hence, $\epsilon > \epsilon^*_2$. Since this is true for all $\epsilon>\hat{\epsilon}$, then we deduce that $\hat{\epsilon} \geq \epsilon^*_2$ and complete the proof.
\end{proof}

We proceed by providing a numerical way to calculate the threshold of a given set of degree distributions. Define
\begin{align} \label{eq:q_L q_J}
q_1(x)\triangleq x\cdot \frac{\Lambda^{(1)}\left( 1- \rho^{(1)}\left(1-x \right)\right)}{\lambda^{(1)}\left( 1- \rho^{(1)}\left(1-x \right)\right)},\quad q_2(x)\triangleq x\cdot \frac{\Lambda^{(2)}\left( 1- \rho^{(2)}\left(1-x \right)\right)}{\lambda^{(2)}\left( 1- \rho^{(2)}\left(1-x \right)\right)},\quad x \in (0,1].
\end{align}

\begin{lemma} \label{lemma:q_L(0)}
$\lim_{x \to 0}q_1(x)=0$.
\end{lemma}

\begin{proof}
See Appendix~\ref{App:q_L(0)}.
\end{proof}

Since $q_1(1)=1$, Lemma~\ref{lemma:q_L(0)} and the intermediate-value theorem imply that for every $w \in (0,1]$, there exists $x\in (0,1]$ such that $q_1(x)=w$. Note that it is not true in general that $\lim_{x \to 0}q_2(x)=0$ (another evidence of asymmetry); this limit may be infinite (for example the case $P_0>0$, $\rho^{(2)}(x)=x^3$ and $\lambda^{(2)}(x)=x^2$). 

\begin{definition} \label{def:q}
For every $y>0$ such that $q_2(y) \leq 1$ define
\begin{align} \label{eq:q def}
q(y)\triangleq\max \{ x \colon q_1(x)=q_2(y)\}.
\end{align}
\end{definition}

\begin{theorem}\label{th:numerical th}
Let $\lambda^{(1)},\rho^{(1)},\lambda^{(2)},\rho^{(2)}$ be degree-distribution polynomials, let $P_0 \in [0,1]$, and let $\epsilon^*_2$ be the decoding threshold of the bilayer ensemble characterized by these degree distributions as $ n\to\infty $. \\
If $P_0=0$ or $\lambda^{(2)}(0) >0$ , then
\begin{align} \label{eq:numerical th1}
\epsilon^*_2 =  \inf\limits_{\substack{y \in (0,1] \\ q_2(y)\leq 1}} \frac{y}{g(1,q(y),y)}.
\end{align}
Else,
\begin{align} \label{eq:numerical th2}
\epsilon^*_2 = \min\left\{ \inf\limits_{\substack{y \in (0,1] \\ q_2(y)\leq 1}} \frac{y}{g(1,q(y),y)},\;\;\frac1{P_0}\cdot \inf_{(0,1]}\frac{x}{\lambda^{(1)}\left(1- \rho^{(1)}\left(1-x \right)\right)}\right\}.
\end{align}
\end{theorem}

\begin{proof}
See Appendix~\ref{App:numerical th}.
\end{proof}
\begin{remark}
	Although the right-hand side of \eqref{eq:numerical th1}
	and the first argument in the $ \min $ operator of \eqref{eq:numerical th2} only have the variable $y$ in them and thus may appear to only depend on layer 2, in fact their values depend on both layers through the function $q(\cdot)$.
	Moreover, it is not clear, apriori, which argument of the $ \min $ operator of \eqref{eq:numerical th2} will be smaller, and one must use the above procedure to calculate the values.
\end{remark}
\begin{example} \label{ex:threshold}
Consider an ensemble characterized by 
\begin{align*}
\lambda^{(1)}(x)=x, \quad \rho^{(1)}(x)=x^9 ,\quad \lambda^{(2)}(x)=0.3396x+0.6604x^4,\quad P_0=0.2667, \quad \rho^{(2)}(x)=x^9 .
\end{align*}
Using \eqref{eq:design rate2} and \eqref{eq:1D th numeric}, the design rate is $R=0.5571$ and the type-1 decoding threshold is $\epsilon_1^*=0.1112$. In view of  \eqref{eq:numerical th2}, the decoding threshold\footnote{better thresholds for that rate are achieved in the next section, and these degree distributions are given to graphically exemplify the results derived so far} when using both layers is $\epsilon_2^*=\min\{0.35,0.4168\}=0.35$. 
Figure~\ref{Fig:DE} illustrates the 2D-DE equations in \eqref{eq:2D DE1}-\eqref{eq:2D DE3} for three different erasure probabilities: $0.33,0.35,0.37$, from left to right, respectively. When the channel's erasure probability is $\epsilon=0.33$, there are no non-trivial $(f,g)$-fixed points -- the decoding process ends successfully, and when $\epsilon=0.37$, there are two $(f,g)$-fixed points, $(0.335,0.3202)$ and $(0.2266,0.1795)$ -- the decoding process gets stuck at $(0.335,0.3202)$. When $\epsilon=0.35=\epsilon_2^*$, there is exactly one $(f,g)$-fixed point at $(0.27,0.237)$, and the dashed and dotted lines are tangent.
\end{example}

\begin{figure}[h!]
\begin{center}
\input{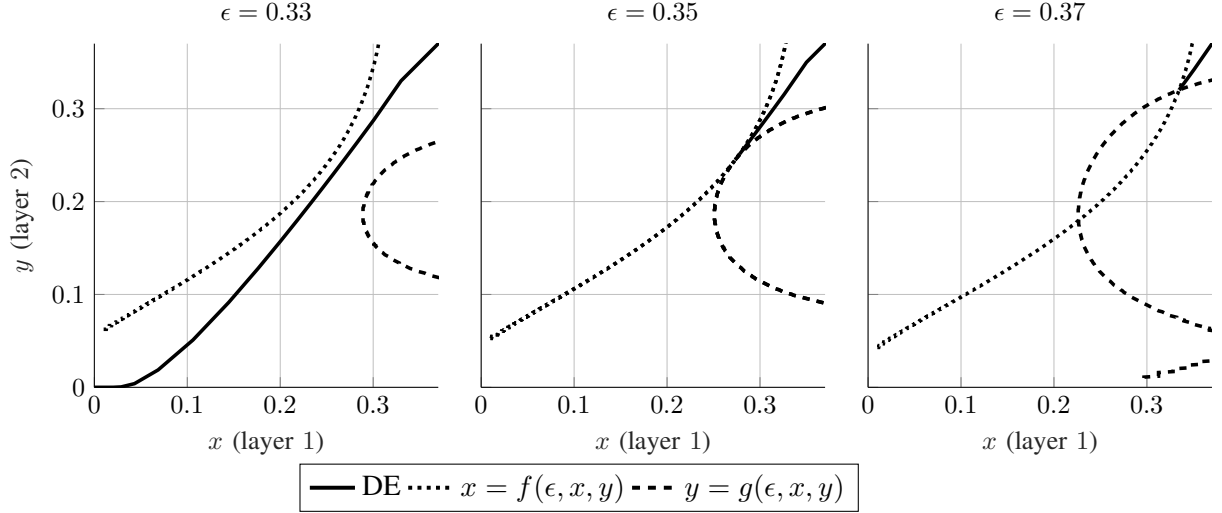}
\end{center}
\caption{\label{Fig:DE}
Illustration of the density-evolution equations in \eqref{eq:2D DE1}-\eqref{eq:2D DE3} for the ensemble in Example~\ref{ex:threshold}, which induce a decoding threshold of $\epsilon^*_2=0.35$. The evolved channel erasure probabilities, from left to right, are $\epsilon=0.33,0.35,0.37$.}
\end{figure}

\section{Code Construction and Approaching Capacity}
\label{Sec:C1}

In this section, we present a code construction, and show how to use this construction to optimally combine two degree distributions (type 1 and type 2) in order to approach capacity.
The inputs for the construction are the desired decoding thresholds, $\epsilon_1$ and $\epsilon_2$, and the outputs are degree-distributions $\left(\lambda^{(1)},\rho^{(1)},\lambda^{(2)},\rho^{(2)},P_0 \right)$ such that
\begin{align*}
\epsilon_1^*\left(\lambda^{(1)},\rho^{(1)} \right)=\epsilon_1 \quad \epsilon^*_2\left(\lambda^{(1)},\rho^{(1)},\lambda^{(2)},\rho^{(2)},P_0 \right)=\epsilon_2.
\end{align*}
The specified parameters $\epsilon_1$ and $\epsilon_2$ can be arbitrarily chosen as fit for the specific application using the codes. $\epsilon_2$ is logically chosen to meet the ``worst-case" noise level in extreme channel instances, while $\epsilon_1$ should specify a lower noise tolerance that is sufficient for a significant fraction of channel instances.  

In principle, setting $P_0=0$, and picking any two LDPC ensembles $\left(\lambda^{(1)},\rho^{(1)}\right)$ and $\left(\lambda^{(2)},\rho^{(2)}\right)$ that induce thresholds $\epsilon_1$ and $\epsilon_2$, respectively, would suffice, but this choice yields poor rates (intuitively, with that choice the type-1 and type-2 codes do not ``cooperate"). 
Another solution is not using the second layer at all, i.e., choosing $\left(\lambda^{(1)},\rho^{(1)}\right)$ that induce a threshold that equals to $\epsilon_2$, and setting $P_0=1$. However, this solution is an undesired overkill since it would miss the opportunity to have a low-complexity decoder for the majority of decoding instances where the erasure probabilities are below $\epsilon_1$.

\begin{definition} \label{def:h_eps}
Let $\left(\lambda^{(1)},\rho^{(1)} \right)$ be type-1 degree-distribution (DD) polynomials, and let $\epsilon_1$ be their decoding threshold.
For $\epsilon \in (\epsilon_1,1)$, let
\begin{enumerate}
\item $h_{\epsilon}(x)=\epsilon\cdot\lambda^{(1)}(1-\rho^{(1)}(1-x))-x,\quad x \in [0,1]$
\item \label{def:x_s}$x_s(\epsilon)=\max\{x \in [0,1] \colon h_\epsilon(x)\geq 0\}$ 
\item $a_s(\epsilon)=\Lambda^{(1)}\left(1- \rho^{(1)}\left(1-x_s(\epsilon) \right)\right)$
\end{enumerate} 
\end{definition}

For every $x \in [0,1]$, $h_\epsilon(x)$ is the erasure-probability change in one BP iteration on the type-1 sub-graph, if the current erasure probability is $x$. By definition, since $\epsilon>\epsilon_1$, the type-1 sub-graph fails to decode, and $h_\epsilon(x)>0$ for some $x \in [0,1]$. In addition, for every $x>\epsilon$, $h_\epsilon(x) < 0$, so $x_s(\epsilon)$ is well defined. Operationally, $x_s(\epsilon)$ is the type-1 erasure probability when the decoder gets stuck (using only type-1 edges). Items 1 and 2 have appeared in \cite{KudRichUrb11}; we add $a_s(\epsilon)$ as a function of  $x_s(\epsilon)$ that encapsulates the erasure probability passed from layer 1 to layer 2.

\subsection{Code Construction}\label{Sub:Construction}
\begin{construction}\label{Construct:LDPCL_BEC}~

\underline{Input:} thresholds $\epsilon_{1}$ and $\epsilon_{2}>\epsilon_1$.

\begin{enumerate}
	\item \label{Item:local dd BEC}Choose any type-1 DD $(\lambda^{(1)},\rho^{(1)})$ that induce a threshold $\epsilon_{1}$.
	\item Calculate $ a_s(\epsilon_2) $.
	\item \label{item:eps_J} Choose any type-2 DD $(\lambda^{(2)},\rho^{(2)})$ that induce a threshold $\epsilon_{2}\cdot a_s(\epsilon_2)$.
	\item Set $ P_0=\epsilon_1/\epsilon_2 $.
\end{enumerate}	

\end{construction}

\begin{remark}
	The main feature of the above construction is that one can use off-the-shelf codes for the two layers and no further optimization is needed. Moreover, if the component codes are efficient (in rate) then so is their combination. We will later investigate how the type-1 and type-2 gaps to capacity affect the overall gap to capacity.
\end{remark}
\begin{theorem}\label{th:C1}
Let $ (\lambda^{(1)},\rho^{(1)},\lambda^{(2)},\rho^{(2)},P_0)$ be degree distributions constructed by Construction~\ref{Construct:LDPCL_BEC}. Then $\epsilon^*_2(\lambda^{(1)},\rho^{(1)},\lambda^{(2)},\rho^{(2)},P_0)=\epsilon_2$.
\end{theorem}

\begin{proof}
From \eqref{eq:1D th numeric} and Theorem~\ref{th:numerical th}, since $P_0=\epsilon_1/\epsilon_2>0$, we have
\begin{align}\label{eq:up bound th C1}
\epsilon^*_2 \leq \frac1{P_0}\cdot \inf_{(0,1]}\frac{x}{\lambda^{(1)}\left(1- \rho^{(1)}\left(1-x \right)\right)} = \frac{\epsilon_1}{P_0} = \epsilon_2.
\end{align}

For the opposite direction, let $\epsilon_1<\epsilon<\epsilon_2$. In view of Theorem~\ref{theorem:fix pt char}, it suffices to show that \eqref{eq:f,g fixed} has no solution for $(x,y)\in (0,1]\times [0,1]$. In view of \eqref{eq:f}, for every $ x \in (0,1]$,
\begin{align} \label{eq:f<x y=0}
f({\epsilon},x,0) &= {\epsilon} \lambda^{(1)}(1-\rho^{(1)}(1-x))P_0 \notag \\
&< P_0 \epsilon_2 \lambda^{(1)}(1-\rho^{(1)}(1-x)) \notag \\
&= \epsilon_1 \lambda^{(1)}(1-\rho^{(1)}(1-x)) \notag \\
&\leq x .
\end{align}
Furthermore, Definition~\ref{def:h_eps} implies that for every $(x,y) \in (x_s({\epsilon}), 1)\times[0,1]$,  
\begin{align}\label{eq:f<x}
f({\epsilon},x,y) 
&= {\epsilon} \lambda^{(1)}(1-\rho^{(1)}(1-x)) \Lambda^{(2)}\left(1- \rho^{(2)}\left(1-y\right)\right) \notag \\
&< x\Lambda^{(2)}\left(1- \rho^{(2)}\left(1-y\right)\right) \notag \\
&\leq x,
\end{align}
and from Lemma~\ref{lemma:monotonicity}, if $(x,y) \in (0,x_s({\epsilon})]\times(0,1]$, 
\begin{align} \label{eq:<a_eps}
g({\epsilon},x,y)
&\leq g({\epsilon},x_s({\epsilon}),y) \notag \\
&= {\epsilon} \lambda^{(2)}(1-\rho^{(2)}(1-y)) \Lambda^{(1)}\left(1- \rho^{(1)}\left(1-x_s({\epsilon})\right)\right) \notag \\
&= {\epsilon} \lambda^{(2)}(1-\rho^{(2)}(1-y)) a_s({\epsilon}).
\end{align}
Since ${\epsilon}<\epsilon_2$, then $\epsilon_2\cdot a_s(\epsilon_2) > {\epsilon}\cdot a_s({\epsilon})$; thus \eqref{eq:<a_eps} yields 
\begin{align} \label{eq:g<y}
g({\epsilon},x,y)< y ,\quad \forall (x,y) \in (0,x_s({\epsilon})]\times(0,1].
\end{align}
Combining \eqref{eq:f<x y=0}, \eqref{eq:f<x}, and \eqref{eq:g<y} implies that \eqref{eq:f,g fixed} has no solution in $(0,1]\times [0,1]$. Thus $\epsilon^*_2 \geq \epsilon$. Since this is true for any ${\epsilon}<\epsilon_2$, we conclude that 
\begin{align*}
\epsilon^*_2 \geq \epsilon_2,
\end{align*}
which combined with \eqref{eq:up bound th C1} completes the proof.
\end{proof}

\begin{remark}
In most cases, it is hard to produce an analytical expression for $x_s(\epsilon)$, but if we limit the type-1 degrees of the ensemble to be small, then a closed-form expression could be derived for $x_s(\epsilon)$ and $a_s(\epsilon)$.
\end{remark}

\begin{example}
Consider type-1 degree distributions taking the form:
\begin{align} \label{eq:low local2}
	\lambda^{(1)}(x)=x,\, \rho^{(1)}(x) = \rho^{(1)}_{2}x+\rho^{(1)}_{3}x^2+\rho^{(1)}_{4}x^3, \quad \rho^{(1)}_{i}\geq 0,\,\sum_{i=2}^4  \rho^{(1)}_{i}=1.
\end{align}  
In view of \eqref{eq:1D th numeric}, for the family of ensembles given in \eqref{eq:low local2}, $\epsilon_1=\frac1{1+\rho^{(1)}_{3}+2\rho^{(1)}_{4}}\;.$
In addition, for every $\epsilon_2 \in (\epsilon_1,1)$,
\begin{align*}
	x_s(\epsilon_2)=\left\{
	\begin{array}{ll}
	1-\tfrac1{\rho^{(1)}_{3}}\left(1/\epsilon_2-1 \right), & \rho^{(1)}_{4}=0 \\
	\tfrac{\rho^{(1)}_{3}+3\rho^{(1)}_{4}-\sqrt{\left( \rho^{(1)}_{3}+\rho^{(1)}_{4}\right)^2+4\rho^{(1)}_{4}\left(1/\epsilon_2-1 \right)}}{2\rho^{(1)}_{4}}, & \rho^{(1)}_{4}>0
	\end{array}
	\right.\;.
\end{align*}
Finally, since $ \Lambda^{(1)}(x)=x^2 $, and in view of Definition~\ref{def:h_eps},
\begin{align*}
a_s(\epsilon_2) = \Lambda^{(1)}(1-\rho^{(1)}(1-x_s(\epsilon_2))) = \left (\frac{x_s(\epsilon_2)}{\epsilon_2}\right )^2\,.
\end{align*}
These closed-form expressions of $ x_s $ and $ a_s $ can be used for constructing a layer-1 code designed for certain parameters $\epsilon_1\;, \epsilon_2$, using a simple optimization of the parameters $\rho^{(1)}_{2}\;,\rho^{(1)}_{3}\;,\rho^{(1)}_{4}$. This optimization maximizes the fraction of bits layer 1 uncovers for layer 2 in case the channel parameter is $ \epsilon_2 $, while guaranteeing its own noise resilience $ \epsilon_1 $ (we omit the details here). 

\end{example}

\subsection{Approaching Capacity} \label{sub:c.a. seq}

In this sub-section, we show how to approach capacity in the bilayer framework, using the construction proposed in sub-section~\ref{Sub:Construction}. More specifically, we upper bound the bilayer additive gap to capacity with a certain linear combination of the gaps to capacity of the two component codes. During the derivation, we refer to $\delta(\lambda,\rho)$ as the additive gap to capacity of the (single layer)  $LDPC(\lambda,\rho)$ ensemble, i.e., $\delta(\lambda,\rho)=1-\epsilon^*(\lambda,\rho)-R(\lambda,\rho)$. Similarly, we define  $\delta\left(\lambda^{(1)},\rho^{(1)},\lambda^{(2)},\rho^{(2)},P_0\right)=1-\epsilon^*_2(\lambda^{(1)},\rho^{(1)},\lambda^{(2)},\rho^{(2)},P_0)-R(\lambda^{(1)},\rho^{(1)},\lambda^{(2)},\rho^{(2)},P_0)$ as the bilayer gap to capacity.

\begin{definition} \label{def:c.a. seq}
	A sequence of bilayer degree distributions $\left\{ \lambda^{(1;k)}(\cdot),\rho^{(1;k)}(\cdot),\lambda^{(2;k)}(\cdot),\rho^{(2;k)}(\cdot),{P_0}_k\right\}_{k \geq 1}$ with associated rates $ \{R_k\}_{k\geq1} $ is said to approach capacity on a BEC with two channel parameters $ 0 < \epsilon_1 < \epsilon_2<1$ if:
	\begin{enumerate}
		\item The threshold of layer 1 approaches $\epsilon_1$ as $ k\to\infty $.
		\item The threshold with both layers approaches $\epsilon_2$ as $ k\to\infty $.
		\item $R_k$ approaches $ 1-\epsilon_2 $ as $k\to\infty $.
	\end{enumerate}
\end{definition}
Note that items~2 and~3 imply that $\lim_{k \to \infty}\delta\left(\lambda^{(1;k)}(\cdot),\rho^{(1;k)}(\cdot),\lambda^{(2;k)}(\cdot),\rho^{(2;k)}(\cdot),{P_0}_k\right)=0$.

\begin{lemma} \label{lemma:gap2cap}
	Let $\left(\lambda^{(1)},\rho^{(1)},\lambda^{(2)},\rho^{(2)},P_0\right)$ be bilayer degree-distribution polynomials constructed according to Construction~\ref{Construct:LDPCL_BEC}, and let $\delta_1 \triangleq \delta(\lambda^{(1)},\rho^{(1)})$ and $\delta_2\triangleq\delta(\lambda^{(2)},\rho^{(2)})$ be the type-1 and type-2 (additive) gaps to capacity, respectively. Then,
	\begin{align} \label{eq:delta}
	\delta\left(\lambda^{(1)},\rho^{(1)},\lambda^{(2)},\rho^{(2)},P_0\right) \leq \delta_1+\delta_2\cdot \left( 1-P_0\right).
	\end{align}
\end{lemma}

\begin{proof}
	See Appendix~\ref{App:Multigap2cap}.
\end{proof}

\begin{remark}
	In principle, the bound in Lamma~\ref{lemma:gap2cap} is not tight since $ a_s(\epsilon_2)\epsilon_2<\epsilon_2 $, and the gap to capacity is smaller. However, since the left-hand side of the last inequality depends on the particular degree distributions used in Construction~\ref{Construct:LDPCL_BEC}, the bound \eqref{eq:delta} has the advantage of applying in full generality. 
\end{remark}

At this point, it should be clear how to construct a capacity-approaching sequence of bilayer ensembles with two thresholds $ 0 < \epsilon_1 < \epsilon_2<1$. Choose any two sequences of ``ordinary" LDPC ensembles $\left\{ \lambda^{(1;k)}(\cdot),\rho^{(1;k)}(\cdot),\right\}_{k \geq 1}$ and $\left\{ \lambda^{(2;k)}(\cdot),\rho^{(2;k)}(\cdot),\right\}_{k \geq 1}$ that achieve capacity on the BEC($\epsilon_1$) and BEC($\epsilon_2$), respectively, and set ${P_0}_k=\left(1-{\epsilon_1}/{\epsilon_2}\right)$, for all $k \geq 1$. Item 1 in Definition~\ref{def:c.a. seq} clearly holds for this  sequence, and in view of Theorem~\ref{th:C1}, item 2 in Definition~\ref{def:c.a. seq} holds as well. Finally, Lemma~\ref{lemma:gap2cap} implies that
\begin{align*}
\lim\limits_{k \to \infty}\delta\left(\lambda^{(1;k)}(\cdot),\rho^{(1;k)}(\cdot),\lambda^{(2;k)}(\cdot),\rho^{(2;k)}(\cdot),{P_0}_k\right)  
&\leq  \lim\limits_{k \to \infty}\delta\left(\lambda^{(1;k)}(\cdot),\rho^{(1;k)}(\cdot)\right)  \\
&\hspace*{4mm}+  \lim\limits_{k \to \infty}\delta\left(\lambda^{(2;k)}(\cdot),\rho^{(2;k)}(\cdot)\right)\left(1-{\epsilon_1}/{\epsilon_2}\right) \\
&=0.
\end{align*} 

\begin{example} \label{ex:c.a.}
	We construct a bilayer capacity-achieving sequence with thresholds $\epsilon_1=0.05$ and $\epsilon_2=0.2$. We set $P_0={\epsilon_1}/{\epsilon_2}=0.25$ and we use the Tornado capacity-approaching sequence \cite{Luby01},
	\begin{align} \label{eq:Tornado}
	\begin{split}
	& \lambda^{(1)}(x)=\frac1{H(D_1)}\sum_{i=1}^{D_1} \frac{x^i}{i}, \qquad
	\lambda^{(2)}(x)=\frac1{H(D_2)}\sum_{i=1}^{D_2} \frac{x^i}{i}, \\
	& \rho^{(1)}(x)=e^{-\alpha_1}\sum_{i=0}^{\infty} \frac{(\alpha_1 x)^i}{i!}, \qquad
	\rho^{(2)}(x)=e^{-\alpha_2}\sum_{i=0}^{\infty} \frac{(\alpha_2 x)^i}{i!},
	\end{split}
	\end{align}
	where $H(\cdot)$ is the harmonic sum, $\alpha_j=\tfrac{H(D_j)}{\epsilon_j}$ , $j\in\{1,2\}$ (the check degree-distribution series are truncated to get degree-distribution polynomials with finite degrees). $D_1$ (resp. $D_2$) controls the type-1 (resp. type-2) gap to capacity $\delta_1$ (resp. $\delta_2$); the bigger it is, the smaller the gap is. 
	
	Figure~\ref{Fig:CA} exemplifies how the sequence 
	$\left\{\lambda^{(1)},\rho^{(1)},\lambda^{(2)},\rho^{(2)},P_0 \right\}$ approaches capacity as $D_1\to \infty,\;D_2\to \infty$: Theorem~\ref{th:C1} implies that for every value of $D_1$ and $D_2$, the global decoding threshold is $\epsilon^*_2\geq0.2$; the type-1 gap to capacity $\delta_1$ and type-2 gap to capacity $\delta_2$ both vanish as $D_1\to \infty$ and $D_2\to \infty$, which in view of \eqref{eq:delta}, implies that the overall gap to capacity $\delta$ vanishes as well.
	In addition, as demonstrated in Figure~\ref{Fig:CA}, $\delta_1$ vanishes much faster with $D_1$ thanks to the lower $\epsilon_1$.
	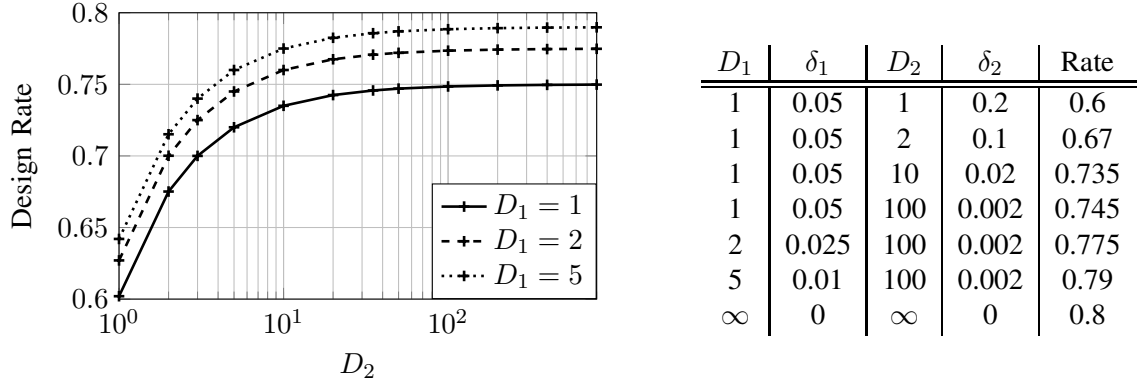
\begin{figure}
		\begin{minipage}{.6\textwidth}
\begin{tikzpicture}

\begin{axis}[%
at={(0,0)},
width=2.5in,
height=1.5in,
scale only axis,
xmode=log,
xmin=1,
xmax=800,
xminorticks=true,
xmajorticks=true,
xlabel={$D_2$},
ymin=0.6,
ymax=0.8,
ylabel={Design Rate},
xmajorgrids,
xminorgrids,
ymajorgrids,
legend entries={$D_1=1$,$D_1=2$,$D_1=5$},
legend style={
	legend cell align=left, 
	align=left, 
	draw, 
	fill=white,
	at={(1,0)},
	anchor=south east}
]

\addplot [line width=1.0pt, mark=+, mark options={solid}]
  table[row sep=crcr]{%
1	0.602021384305841\\
2	0.675124444189399\\
3	0.700020893034882\\
5	0.720001982085271\\
10	0.735000072202418\\
20	0.7425000026321\\
35	0.745714286073054\\
50	0.747000000232104\\
100	0.74850000020694\\
200	0.749250000206142\\
400	0.749625000206116\\
800	0.749812500206115\\
};
\addplot [dashed,line width=1.0pt, mark=+, mark options={solid}]
  table[row sep=crcr]{%
1	0.627021384099733\\
2.0	0.70012444398329\\
3.0	0.725020892828773\\
5.0	0.745001981879163\\
10 	0.76000007199631\\
20 	0.767500002425992\\
35 	0.770714285866945\\
50 	0.772000000025996\\
100	0.773500000000832\\
200	0.774250000000033\\
400	0.774625000000008\\
800	0.774812500000007\\
};
\addplot [dotted,line width=1.0pt, mark=+, mark options={solid}]
  table[row sep=crcr]{%
1  0.642021384099757\\
2  0.715124443983314\\
3  0.740020892828797\\
5  0.760001981879187\\
10 0.775000071996334\\
20 0.782500002426016\\
35 0.785714285866969\\
50 0.78700000002602\\
100 0.788500000000856\\
200 0.789250000000057\\
400 0.789625000000032\\
800 0.789812500000031\\
};
\end{axis}
\end{tikzpicture}%

\end{minipage}
\begin{minipage}{.4\textwidth}
	\begin{tabular}{ c|c|c|c|c}
		$D_1$&$\delta_1$&$D_2$&$\delta_2$&Rate    \\ \hline \hline
		1    &  0.05    &  1  &  0.2     & 0.6    \\
		1    &  0.05    &  2  &  0.1     & 0.67   \\
		1    &  0.05    &  10 &  0.02    & 0.735  \\
		1    &  0.05    &  100&  0.002   & 0.745  \\
		2    &  0.025   &  100&  0.002   & 0.775  \\
		5    &  0.01    &  100&  0.002   & 0.79   \\
		$\infty$&     0    &$\infty$&  0    & 0.8    
	\end{tabular}
	
\end{minipage}
		\caption{\label{Fig:CA} Left side: design rate of a the bilayer ensembles constructed by Construction~\ref{Construct:LDPCL_BEC} with two Tornado component codes from Example~\ref{ex:c.a.}. The horizontal axis $D_2$ and the different plots $ D_1 $ are the degree of the polynomials $ \lambda^{(2)}(\cdot)$ and $ \lambda^{(1)}(\cdot) $, respectively. Right side: type-1 gap to capacity $\delta_1$, type-2 gap to capacity $\delta_2$, and rate of the bilayer capacity-achieving sequence from Example~\ref{ex:c.a.}.}
	\end{figure}
	
\end{example}

\begin{remark}
	Figure~\ref{Fig:CA} (in particular the second from bottom row in the table) shows the advantage of the bilayer scheme: one can get very close to capacity with type-1 ensembles that are extremely low complexity thanks to their low $D_1$ values. 
\end{remark}

\begin{remark}\label{remark:complexity}
	As mentioned in Section~\ref{Sec:Intro}, the complexity advantage of decoding bilayer LDPC codes using only layer 1 over decoding ordinary LDPC codes comes from the fact that since layer 1 is designed for lower noise levels, for the same gap to capacity its node degrees are lower than those of the ordinary LDPC. For the BEC, counting edges in the Tanner graph is a first-order approximation of the decoding complexity. We now perform a comparison between two codes: Code~1 is a bilayer LDPC code, and Code~2 an ordinary LDPC code. Code~1 is constructed by Construction~\ref{Construct:LDPCL_BEC} for erasure levels $ 0<\epsilon_1<\epsilon_2<1 $, and Code~2 is designed for the worst-case channel parameter $\epsilon_2$. In particular, we  take Code~1 degree distributions $ \left( \lambda^{(1)},\rho^{(1)},\lambda^{(2)},\rho^{(2)},P_0\right ) $ from Example~\ref{ex:c.a.} with $D_1=2,D_2=10$. Code~2's degree distributions are $\lambda^{(2)},\rho^{(2)}$.
	The variable-node type-1 and type-2 average degrees in Code~1 (bilayer) are denoted by $ d_1  $ and $d_2 $, respectively. The variable-node average degree in Code~2 is denoted by $ d $. It is known that if the variable-node degree distribution is given by $ \lambda(\cdot) $, then their average degree is given by $ 1/\!\int_0^1 \lambda  $. Hence 
	\begin{align*}
	&d_1 =\left (\int_0^1\lambda^{(1)}(x)\mathrm{d}x\right )^{-1}=2.25,\\
	&d_2 =(1-P_0)\cdot\left (\int_0^1\lambda^{(2)}(x)\mathrm{d}x\right )^{-1}=2.41\\
	&d   =\left (\int_0^1\lambda^{(2)}(x)\mathrm{d}x\right )^{-1}=3.22.
	\end{align*} 
	The complexity reduction when decoding layer 1 is $30.16\%$ compared to the ordinary LDPC code. When decoding both layers we pay with an increase of $45.37\%$ in the average degree, but this applies only to the 2-layers decoder which we assume to be used either infrequently (if most decoding instances have erasure rates below $\epsilon_1$) or by nodes where computational parsimony is less critical (compared to nodes performing layer-1 decoding).
\end{remark}

\section{General Multi-Layer Construction}\label{Sec:Multi}
In this section, we show how to generalize Construction~\ref{Construct:LDPCL_BEC} to more than two layers. This generalization is motivated, for example, by multi-block coding \cite{RamCassuto18a} with a hierarchical structure where a number of sub-blocks are joined to form a super-block, and a number of super-blocks are joined further, etc. Another example is a multiple-relay channel, with a source, $ L-1 $ relays, and a destination (see \cite{RazaYu07} for the relay channel with $L=2$). The source sends a message to all relays and destination, and for every $ i\in\{1,2,\ldots,L-1\} $, the $ i $-th relay decodes its incoming message and forwards parity bits to relays $ j\in\{i+1,\ldots,L-1\} $, and to the destination.

The advantage of the framework developed in this paper toward the multi-layer extension is that the number of parameters of the ensemble grows linearly with $L$. An extension of \cite{RazaYu07} to multi-layer codes through multi-variate DD polynomials would imply exponential growth of the number of ensemble parameters.  

\subsection{Code Structure \& Density Evolution}
Let $ L>1 $ be an integer. The $L$-layer ensemble is characterized by the block length $ n $, and $ L $ degree-distributions polynomials $ \left \{\left ( \Lambda^{(i)}(\cdot),\Omega^{(i)}(\cdot)\right )\right \}_{i=1}^L $. Each variable node has $ L $ types of edges emanating from it with degrees specified by $ \{\Lambda^{(i)}(\cdot)\}_{i=1}^L $, and check nodes are divided into $ L $ types with degrees specified by $ \{\Omega^{(i)}(\cdot)\}_{i=1}^L $ where check nodes can connect only to edges of the same type.
For every layer $i\in\{1,2,\ldots,L\}  $, we denote by $ P_0^{(i)}=\Lambda^{(i)}(0)$ the fraction of variable-nodes with no type-$i$ edges. From edge perspective  the degree-distribution polynimals are given by  $ \left \{\left ( \lambda^{(i)}(\cdot),\rho^{(i)}(\cdot), P_0^{(i)}\right )\right \}_{i=1}^L $. 
Since layer $1$ should have a positive threshold, then we require that $  P_0^{(1)}=0 $. For $ i\in\{2,3,\ldots,L\} $ we allow $  P_0^{(i)}>0 $.

The generalization of the density-evolution equations in \eqref{eq:2D DE1}--\eqref{eq:2D DE3} for the multi-layer ensemble are given by
\begin{align}\label{eq:multiDE}
\begin{split}
x^{(i)}_l(\epsilon) &= \epsilon \cdot \lambda^{(i)}\left(1- \rho^{(i)}\left(1-x^{(i)}_{l-1}(\epsilon) \right)\right)\cdot \prod_{j\neq i}\Lambda_j\left(1- \rho_j\left(1-x^{(j)}_{l-1}(\epsilon) \right)\right)\\
&\triangleq f_i\left ( \epsilon,\mathbf{x}^{(L)}_{l-1}(\epsilon)\right ), \quad l \geq 0,\quad \forall 1\leq i \leq L,\\
{x}^{(i)}_{-1}(\epsilon)&=1,\quad \forall 1\leq i \leq L,
\end{split}
\end{align}
where for every $ i\in\{1,2,\ldots,L\} $, $ x^{(i)}_l(\epsilon) $ is the probability that a type-$i$ edge carries a variable-to-check erasure messages after $l$ BP iterations over the BEC($\epsilon$), and $ \mathbf{x}^{(i)}_{l}(\epsilon)=\left ( x^{(1)}_l(\epsilon),x^{(2)}_l(\epsilon),\ldots,x^{(i)}_l(\epsilon)\right ) $. In what follows, we omit $ \epsilon $ from $\mathbf{x}^{(i)}_l(\epsilon)   $, and for brevity we re-write \eqref{eq:multiDE} as $ \mathbf{x}^{(i)}_l = F^{(i)}\left (\epsilon,\mathbf{x}^{(L)}_{l-1}\right )  $ where $ F^{(i)}\colon [0,1]^{i+1}\to [0,1]^L $ encapsulates the first $ i $ density-evolution equations. 
\subsection{Code Construction}
\begin{construction}\label{Const:multi}~
	\begin{itemize}
		\item \underline{Input:} thresholds $0<\epsilon_{1}<\epsilon_{2}<\dots<\epsilon_{L}<1$.
		
		\item \underline{Output:} degree distributions  $ \left \{\left ( \lambda^{(i)}(\cdot),\rho^{(i)}(\cdot), P_0^{(i)}\right )\right \}_{i=1}^L $ such that for every $ i\in \{1,2,\ldots,L\} $, the decoding threshold of the first $ i $ layers equals $\epsilon_i $.
	\end{itemize}

	\begin{enumerate}
		\item Choose any degree distributions $(\lambda^{(1)},\rho^{(1)})$ such that $\epsilon^*(\lambda^{(1)},\rho^{(1)})=\epsilon_{1}$.
		\item For each $ i\in\{1,2,\ldots,L-1\} $ do:
		\begin{enumerate}
			\item \label{it:fixpoint}Calculate $ \left (x_s^{(1)},x_s^{(2)},\ldots,x_s^{(i)}\right )$ as the largest (element-wise) fixed point of 
			
			$F^{(i)}\left (\epsilon_{i+1},x^{(1)},x^{(2)},\ldots,x^{(i)},1,1,\ldots,1\right )  $.
			\item \label{it:as} Calculate $ a_s^{(i)}=\prod_{j=1}^i \Lambda_j\left (1-\rho_j\left (1-x_s^{(j)}\right )\right )$ (quantifies the amount by which layers $ 1,2,\ldots,i $ help when decoding  $ i+1 $ layers).
			\item \label{it:anyDD} Choose any degree distributions $(\lambda_{i+1},\rho_{i+1})$ that induce a threshold $\epsilon_{i+1}\cdot a_s^{(i)}$.
			\item Set $ P_0^{(i+1)}=\epsilon_i/\epsilon_{i+1} $.
		\end{enumerate}
	\end{enumerate}
\end{construction}

\begin{lemma} \label{lemma:Multigap2cap}
	Let $ \delta $ be the additive gap to capacity of the above ensemble. Then,
	\begin{align} \label{eq:Multidelta}
	  \delta \leq \sum_{i=1}^L \delta_i\cdot(1-P_0^{(i)}),
	\end{align}
	where $\delta_i  $ is the individual additive gap to capacity of the $ i $th layer $ (\lambda^{(i)},\rho^{(i)}) $.
\end{lemma}
\begin{proof}
	See Appendix~\ref{App:Multigap2cap}.
\end{proof}
The advantage of the design approach for suggested in this paper is made more prominent in view of Construction~\ref{Const:multi}. For each layer we need only to calculate steps \ref{it:fixpoint}--\ref{it:as}, and then choose any code that meets the criteria in step \ref{it:anyDD}. On the other hand, in the construction of \cite{RazaYu07}, linear programming is used to optimize the \emph{product} degree distribution of all layers; thus the complexity of Construction~\ref{Const:multi} is much smaller.

\section{Reducing The Number of Type-2 Iterations}
\label{Sec:N_JI}
We now return to the specific case of two layers.

It has not been emphasized earlier in the paper, but in practical settings, the type-1 and type-2 decoding iterations may be very different in terms of cost. For example, the hardware that implements the layer-2 checks may be more costly to operate due to higher code complexity. 
That means that even when decoding the two layers, we would like to reduce the number of layer-2 iterations. 
We define the number of layer-2 iterations performed during decoding by $N_2$, where a layer-2 decoding iteration is a round of variable-to-type-2-check messages and type-2-check-to-variable messages.
Ideally, the decoder successfully decodes on the type-1 sub-graph, and no type-2 iterations are needed ($N_{2}=0$); in the asymptotic regime, this happens when the fraction of erased bits $ \epsilon $ is equal or less than the type-1 threshold, i.e., $\epsilon \leq \epsilon_1$. However, if $\epsilon>\epsilon_1$, then at least one type-2 iteration is necessary ($ N_{2}\geq 1$). 

In this section, we suggest a scheduling scheme for updating layer 2 during the decoding of the entire graph. We prove that our scheduling scheme is optimal in the sense of minimizing $N_{2}$. 
It is known that there is a trade-off between rate and the number of decoding iterations (see, for example, \cite[Table \Romannum{3}]{WangRangWesel19}). We extend this observation and study how the parameters of the type-1 and type-2 degree distributions affect $ N_{2}$ when using the optimal scheduling scheme.
Note that our notion of scheduling differs from the standard meaning of scheduling algorithms for iterative-decoding (see \cite{ZhangFossorier02, XiaoBeni04, CasGriotWesel07}). We consider scheduling of type-2 decoding iterations, while previous work considered the order of message passing between nodes in the Tanner graph.

\subsection{An $ N_{2}$-optimal scheduling scheme}
\label{sub:opt schedule}
Recall that in the bilayer density-evolution equations, \eqref{eq:2D-DE1} and \eqref{eq:2D-DE2} express a type-1 and a type-2 iteration, respectively.
A scheduling scheme prescribes decoder access to the type-2 check nodes in only part of the iterations, and thus replaces \eqref{eq:2D-DE2} with 
\begin{align}\label{eq:schedule}
y_l = \left\{
\begin{array}{ll}
 g\left(\epsilon,x_{l-1},y_{l-1}\right) & l \in A \\
y_{l-1} & l \notin A
\end{array}
\right.
\end{align}
for some $A \subseteq \mathbb{N}$ representing the iteration numbers where type-2 checks are accessed; in this case we have
$  N_{2}=\vert A \vert.$
Since Lemma~\ref{lemma:mono of DE} (monotonicity) still holds when \eqref{eq:2D-DE2} is replaced with \eqref{eq:schedule}, the limits $\lim_{l \to \infty}x_l$ and $\lim_{l \to \infty}y_l$ exist for every scheduling scheme.

Given type-1 and type-2 degree distributions, a scheduling scheme is called \emph{valid} if for every $\epsilon$ less than the ensemble's threshold $ \epsilon_2$, $\lim_{l \to \infty}x_l(\epsilon)=0$ (successful decoding). Our goal is to find an \emph{optimal} scheduling scheme: a valid scheduling scheme that minimizes $ N_{2}$. For example, if $A=\emptyset$, then $ N_{2}=0$ but $\lim_{l \to \infty}x_l(\epsilon)>0$ if $\epsilon \in (\epsilon_1, \epsilon_2)$; thus, the scheduling scheme is not valid. If, on the other hand, type-2 checks are accessed in every iteration (as assumed in Sections~\ref{Sec:2D DE}--\ref{Sec:C1}), then the scheduling scheme is valid, but $N_{2}$ equals the total number of iterations, which is the worst case. We do not require the scheduling scheme to be pre-determined, and it can use ``on-line" information about the decoding process. For example, it can use the current fraction of erasure messages or the change in this fraction between two consecutive iterations. 
\begin{definition} \label{def:eff epsilon}
Let $(\Lambda^{(2)},\rho^{(2)})$ be type-2 degree-distribution polynomials, let $\epsilon\in (0,1)$ be the erasure probability of a BEC, and let $y \in [0,1]$ be an instantaneous erasure probability of a type-2 edge. We define the effective erasure probability from layer 1's perspective as
\begin{align} \label{eq:eff loc eps}
\epsilon^{(1)}_{\mathrm{eff}}(\epsilon,y) = \epsilon \cdot \Lambda^{(2)}\left(1-\rho^{(2)}(1-y)\right).
\end{align}
\end{definition}
In view of \eqref{eq:f} and \eqref{eq:eff loc eps}, we have
\begin{align} \label{eq:f with eff}
x_l = f(\epsilon,x_{l-1},y_{l-1})=\epsilon^{(1)}_{\mathrm{eff}}(\epsilon,y_{l-1}) \lambda^{(1)}(1-\rho^{(1)}(1-x_{l-1})).
\end{align}
$\epsilon^{(1)}_{\mathrm{eff}}(\epsilon,y_{l-1}) $ takes the role of $\epsilon$ when layer 1 is viewed as a standard LDPC code, hence the term ``effective erasure probability from layer 1's perspective".

Our proposed scheduling scheme is parameterized by $\eta>0$, and is given by
\begin{align}\label{eq:opt schedule}
\begin{split}
&x_l=f(\epsilon,x_{l-1},y_{l-1}) \\
&y_l = \left\{
\begin{array}{ll}
g\left(\epsilon,x_{l-1},y_{l-1}\right) & \left| x_{l-2}-x_{l-1}\right|\leq \eta \text{ and } \epsilon^{(1)}_{\mathrm{eff}}(,\epsilon,y_{l-1})\geq \epsilon_1 \\
y_{l-1} & \text{else}
\end{array}
\right. .
\end{split} 
\end{align}
\begin{lemma} \label{lemma:valid opt schedule}
For every $\eta>0$, the scheduling scheme described in \eqref{eq:opt schedule} is valid.
\end{lemma}
\begin{proof}
See Appendix~\ref{App:valid opt schedule}.
\end{proof}
Note that if $\eta=0$, the scheduling scheme described in \eqref{eq:opt schedule} is not valid. However, since type-1 iterations have zero cost in our model, we can assume that we can apply arbitrarily many type-1 iterations to get arbitrarily close to $\eta=0$. Numerical simulations show that $\eta=10^{-4}$ suffices for achieving minimal $N_{2}$. For the following analysis we will assume that $\eta=0$, and that the scheduling scheme is still valid. In this scheduling scheme, the decoder tries to decode the type-1 sub-graph until it gets ``stuck", which refers to not being able to reduce the erasure probability while it is still strictly greater than zero. This happens first when $x_{l_1}=x_s(\epsilon)$, for some iteration $l_1$, where $x_s(\epsilon)$ is given in Definition~\ref{def:h_eps}. So, in the first type-2 update we have
\begin{align*}
\begin{array}{ll}
x_{l_1} = x_s(\epsilon) ,& y_{l_1} =1 \\
x_{l_1+1} = x_s(\epsilon) ,& y_{l_1+1} =g(\epsilon,x_s(\epsilon),1). 
\end{array}
\end{align*}
In view of \eqref{eq:f with eff}, the type-1 sub-graph now ``sees" $\epsilon^{(1)}_{\mathrm{eff}}(\epsilon,y_{l_1+1})<\epsilon$ as an effective erasure probability, and it can continue the decoding algorithm without accessing type-2 edges. It may get ``stuck" again and another type-2 update will be invoked; this procedure continues until $\epsilon^{(1)}_{\mathrm{eff}}(\epsilon,y_{l_p+1}) < \epsilon_1$ in the $p$-th (and last) update, which enables successful decoding (i.e., $N_{2}=p$). In general, let $\{l_k\}_{k=1}^{N_{2}}$ be the type-2 update iterations of the scheduling scheme described above and let $\varepsilon_k $ be the effective erasure probability from layer-1 perspective between type-2 updates $k-1$ and $k$. Then,  
\begin{subequations}
\begin{align}
\label{eq:y_1}
&y_{l_1}=1,\;\varepsilon_1=\epsilon,\;x_{l_1}=x_s(\epsilon), \\
\label{eq:y_l_k}
&y_{l_k}=g\left (\epsilon,x_{l_{k-1}},y_{l_{k-1}}\right ),\quad 2\leq k \leq N_{2}, \\
\label{eq:eps_k}
&\varepsilon_k=\epsilon^{(1)}_{\mathrm{eff}}(\epsilon,y_{l_k}),\quad 2\leq k \leq N_{2}, \\
\label{eq:x_l_k}
&x_{l_k}=x_s(\varepsilon_k),\quad 2\leq k \leq N_{2},
\end{align}
\end{subequations}
where
\begin{align} \label{eq:eps_k mono}
\begin{split}
\epsilon=\varepsilon_1 > \varepsilon_2 > \ldots > \varepsilon_{N_{2}-1} \geq  \epsilon_1 >  \varepsilon_{N_{2}}.
\end{split}
\end{align}
\begin{lemma} \label{lemma:opt schedule}
The scheduling scheme described above is optimal. 
\end{lemma}
\begin{proof}
See Appendix~\ref{App:opt schedule}.
\end{proof}
We assume from now on that the decoder applies the optimal scheduling scheme suggested above.

\subsection{The Rate-vs.-$N_{2}$ Trade-Off}

We will now see that the smaller the gap to capacity is, the higher $N_{2}$ is; therefore, to decrease $N_{2}$ we have to pay with rate, and there are several ways to do so. In this section we study how the parameters of the component layers affect $N_{2}$. In particular, we focus on how the type-1 and type-2 additive gaps to capacity $\delta_1$ and $\delta_2$, receptively, affect $N_{2}$.  

It is well known that if $\left\{ \lambda^{(k)},\rho^{(k)}\right\}^{\infty}_{k=1}$ is a (ordinary) capacity-approaching sequence for the BEC($\epsilon$), then 
\begin{align}\label{eq:c.a. known}
\lim_{k \to \infty} \epsilon\lambda^{(k)}(1-\rho^{(k)}(1-x))=x,\quad x \in [0,\epsilon]
\end{align}
(see \cite{Shok01}).
This leads to the following lemma.
\begin{lemma} \label{lemma:x_s c.a.}
Let  $0<\epsilon_1<\epsilon<1$, and let $\left\{ \lambda^{(1;k)},\rho^{(1;k)}\right\}^{\infty}_{k=1}$ be a capacity-approaching sequence for the BEC($\epsilon_1$). Then,
\begin{align*}
x_s(\epsilon)\triangleq\lim_{k \to \infty}x^{(k)}_s(\epsilon)=\epsilon,
\end{align*}
where $x^{(k)}_s(\epsilon)$ corresponds to Definition~\ref{def:h_eps} with $\left(\lambda^{(1;k)},\rho^{(1;k)}\right)$.
\end{lemma}

\begin{proof}
See Appendix~\ref{App:x_s c.a.}.
\end{proof}

Lemma~\ref{lemma:x_s c.a.} asserts that if the type-1 degree-distribution polynomials imply a threshold $\epsilon_1$ and a design rate that is very close to capacity ($1-\epsilon_1$), and the channel erasure probability $\epsilon$ is greater than $\epsilon_1$, then the BP decoding algorithm on the type-1 sub-graph gets ``stuck" immediately after correcting only a small fraction of the erasures. This leads, in view of \eqref{eq:y_l_k}, to a small change in the erasure-message probability on the type-2 update, which in turn yields a minor progress in the type-1 side. Therefore, choosing close to capacity \emph{type-1} degree-distribution polynomials implies high $N_{2}$.
Another consequence of \eqref{eq:c.a. known} is that the change in the erasure-message probability in one iteration of the BP decoding algorithm is small. Thus, close to capacity \emph{type-2} degree-distribution polynomials yield high $N_{2}$, regardless of the \emph{type-1} degree-distribution polynomials.

\begin{example} \label{ex:N_JI}
Let $\epsilon_1=0.05$ and $\epsilon_2=0.2$. We use the capacity-achieving sequence given in Example~\ref{ex:c.a.}. 
A computer program simulated \eqref{eq:y_1}-\eqref{eq:x_l_k} with $ \epsilon=0.1998 $ ($ 99.9\% $ of $ \epsilon_2 $) and degree distributions from \eqref{eq:Tornado} with the same values of $D_1$ and $D_2$ as in Figure~\ref{Fig:CA}. The results are presented in Figure~\ref{Fig:EL05P1EG2}. Figure~\ref{Fig:CA} and Figure~\ref{Fig:EL05P1EG2} exemplify the trade-off between rate and $N_{2}$: when the ensemble is close to capacity with $\delta_1=10^{-2},\delta_2=2.5\cdot10^{-4}$ ($ D_1=5,D_2=800,R=0.79$), we get $N_{2}=570$, and to reduce $N_{2}$ we have to pay with rate. However, there are several ways to do so. For example, changing the type-1 gap to $\delta_1=5\cdot10^{-2}$ ($ D_1=1 $) while the type-2 gap stays $\delta_2=2.5\cdot10^{-4}$ (red circle labeled $A$ in the plot) yields $R=0.75$ and $N_{2}=26$, and changing the type-1 and type-2 gap to $\delta_1=2.5\cdot10^{-2}$ ($ D_1=2 $) and $\delta_2=4\cdot10^{-2}$ ($ D_2=5 $), respectively (blue circle labeled $B$), yields the same $R=0.75$ but a smaller $N_{2}=11$.

\begin{figure}[!h]
\begin{center}
\begin{tikzpicture}

\begin{axis}[%
width=5in,
height=2in,
scale only axis,
xmode=log,
xmin=0.0002,
xmax=0.2,
xminorticks=true,
xlabel={$\delta_2$},
ymin=0,
ymax=600,
ylabel={$N_{2} $},
axis background/.style={fill=white},
xmajorgrids,
xminorgrids,
ymajorgrids,
legend style={
	legend cell align=left, 
	align=left, 
	draw, 
	fill=white,
	at={(1,1)},
	anchor=north east}
]
\addplot [dotted, line width=1.0pt,mark=o,mark options={solid}]
table[row sep=crcr]{%
	0.197254817033647	7\\
	0.0998340372099368	11\\
	0.0666388095230959	15\\
	0.0399973574944374	22\\
	0.0199999040049642	39\\
	0.0099999967653871	73\\
	0.00571428551077802	120\\
	0.0039999999653707	160\\
	0.00199999999893696	256\\
	0.000999999999992784	363\\
	0.00050000000003414	469\\
	0.000250000000009409	568\\
};
\addlegendentry{$\delta_1=0.01$}
\addplot [dashed, line width=1.0pt,mark=o,mark options={solid}]
table[row sep=crcr]{%
	0.197254817033647	5\\
	0.0998340372099368	7\\
	0.0666388095230959	8\\
	0.0399973574944374	11\\
	0.0199999040049642	16\\
	0.0099999967653871	25\\
	0.00571428551077802	34\\
	0.0039999999653707	40\\
	0.00199999999893696	56\\
	0.000999999999992784	73\\
	0.00050000000003414	90\\
	0.000250000000009409	107\\
};
\addlegendentry{$\delta_1=0.025$}

\addplot [line width=1.0pt,mark=o]
  table[row sep=crcr]{%
0.197254817033647	4\\
0.0998340372099368	5\\
0.0666388095230959	5\\
0.0399973574944374	7\\
0.0199999040049642	8\\
0.0099999967653871	11\\
0.00571428551077802	13\\
0.0039999999653707	14\\
0.00199999999893696	17\\
0.000999999999992784	20\\
0.00050000000003414	23\\
0.000250000000009409	26\\
};
\addlegendentry{$\delta_1=0.05$}

\node [circle,fill=red,scale=0.3] (a) at (axis cs:0.000250000000009409,26) {};

\node(A)[above right=16mm,inner sep=0] at (axis cs:0.000250000000009409,26) {\textcolor{red}{A}};
\draw[thick,->,>=latex,red] (A)--(a);

\node [circle,fill=blue,scale=0.3] (b) at (axis cs:0.0399973574944374,11) {};

\node(B)[above right=16mm,inner sep=0] at (axis cs:0.0399973574944374,11) {\textcolor{blue}{B}};
\draw[thick,->,>=latex,blue] (B)--(b); 

\end{axis}

\end{tikzpicture}%
\caption{\label{Fig:EL05P1EG2}
Plot of $N_{2}$ as a function of the type-2 additive gap to capacity $\delta_2$ for different values of the type-1 additive gap to capacity $\delta_1$. The simulated erasure rate is $ \epsilon=0.1998 $ ($ 99.9\% $ of $ \epsilon_2=0.2 $).}
\end{center}
\end{figure}
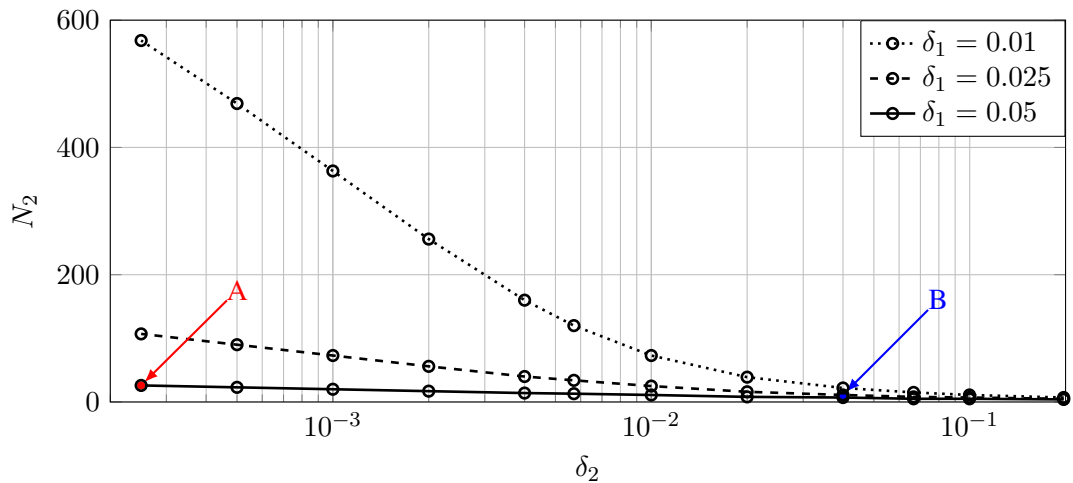
\end{example}

\section{Summary}
\label{Sec:sum}
This paper develops new tools for the construction and analysis of bilayer (and more generally multi-layer) LDPC code ensembles. 
In particular, we derived a code-analysis framework which resulted in a simple way to construct capacity-achieving sequences.
Our design approach lends itself well for an extension to multi-layer code construction without complexity blow-up.
We showed that using this design framework can give codes that enjoy low-complexity layer-1 decoding in low error rates, while still having small gaps to capacity for decoding layers 1+2 in high error rates. Since in some applications it is of interest that the decoding algorithm will avoid layer-2 messages as much as possible, we studied another trade-off regarding the number of layer-2 iterations and the gap to capacity. 

Interesting future work includes combining the asymptotic design techniques with finite-block design techniques for the BEC and other channels. In addition, one can study generalizations of the decoding modes in the $L$-layer framework. Instead of considering $L$ modes: layer 1, layers 1+2,$\ldots$, layers 1+2+$\dots$+$L$ , one can consider other subsets, for example layers 2+3+5+8. In this case, it would be interesting to use the design approach suggested in this paper to optimize the rate and decoding thresholds.

\section{Acknowledgment}
We thank the AE and anonymous referees of a previous version of this paper for valuable comments that improved the presentation considerably.

\appendices 
\section{Proof of Lemma~\ref{lemma:fix pt}} \label{App:fix pt}
\begin{enumerate}
\item Assume that $x=0$. Since $\lambda^{(1)}(0)=0$, \eqref{eq:g} implies that $y=g(\epsilon,0,y)=0$. Moreover, if $y=0$ and $P_0=0$, then \eqref{eq:f} yields $x=f(\epsilon,x,0)=0$. Finally, if $y=0$ and $\lambda^{(2)}(0)>0$, then \eqref{eq:g} implies that $0=\lambda^{(1)}(1-\rho^{(1)}(1-x))$, hence $\rho^{(1)}(1-x)=1$ and $x=0$. 
\item Follows immediately from \eqref{eq:f}, \eqref{eq:g} and \eqref{eq:f,g fixed}.
\item We prove \eqref{eq:no jumps} by a mathematical induction. For $l=0$, \eqref{eq:no jumps} holds due to Item~\ref{item: fix pt lemma1} and the fact that $x_{0} =y_{0}=\epsilon $. Assume correctness of \eqref{eq:no jumps} for some $l \geq 0$ and consider iteration $l+1$. In view of Lemma~\ref{lemma:monotonicity}, \eqref{eq:2D-DE1}--\eqref{eq:2D-DE3} and the induction assumption, it follows that
\begin{align} \label{eq:proof fix pt lemma}
\begin{split}
&x_{l+1} = f \left(\epsilon,x_{l},y_{l}\right) \geq f(\epsilon,x,y)= x , \\
&y_{l+1} = g\left(\epsilon,x_{l},y_{l}\right) \geq g(\epsilon,x,y)= y.
\end{split}
\end{align}
This prove correctness of \eqref{eq:no jumps} for $l+1$ and by mathematical induction proves \eqref{eq:no jumps} for all $l \geq 0$.
\end{enumerate}

\section{Proof of Lemma~\ref{lemma:q_L(0)}} \label{App:q_L(0)}
Let $I=\min \{i \colon \Lambda^{(1)}_{i}>0\}$ be the first non-zero coeficient of $ \Lambda^{(1)}(\cdot) $, and let $C=\tfrac{\mathrm{d}}{\mathrm{d}u}\Lambda^{(1)}(u)\big|_{u=1}$. 
Clearly, $I \geq 2$. 
Since $ \lambda^{(1)}(u)=\frac1C\cdot{\tfrac{\mathrm{d}}{\mathrm{d}u}\Lambda^{(1)}(u)} $, then
\begin{align}
\lim\limits_{u \to 0}\frac{\Lambda^{(1)}(u)}{\lambda^{(1)}(u)}
&=\lim\limits_{u \to 0}\frac{\Lambda^{(1)}(u)}{\tfrac{\mathrm{d}}{\mathrm{d}u}\Lambda^{(1)}(u)}\cdot C\notag\\
&=C\cdot\lim\limits_{u \to 0}\frac{\sum_{i\geq I}\Lambda^{(1)}_{i} u^i}{\sum_{i\geq I}i\Lambda^{(1)}_{i} u^{i-1}} \notag\\ 
&=C\cdot\lim\limits_{u \to 0}\frac{u^I\sum_{i\geq I}\Lambda^{(1)}_{i} u^{i-I}}{u^{I-1}\sum_{i\geq I}i\Lambda^{(1)}_{i} u^{i-I}} \notag\\
&=C\cdot\frac{\Lambda^{(1)}_{I} }{I\Lambda^{(1)}_{I} }\cdot\lim\limits_{u \to 0}u \notag\\
&=0.
\end{align}
Further, let $u(x)=1-\rho^{(1)}(1-x)$ and note that $\lim_{x\to 0} u(x)=0$. Thus,
\begin{align*}
\lim_{x \to 0} x \cdot \frac{\Lambda^{(1)}\left( 1- \rho^{(1)}\left(1-x \right)\right)}{\lambda^{(1)}\left( 1- \rho^{(1)}\left(1-x \right) \right)} =\lim_{x \to 0}x \cdot \lim_{x \to 0}\frac{\Lambda^{(1)}(u(x))}{\lambda^{(1)}(u(x))}=\lim_{x \to 0}x \cdot \lim_{u \to 0}\frac{\Lambda^{(1)}(u)}{\lambda^{(1)}(u)} = 0.
\end{align*}

\section{Proof of Theorem~\ref{th:numerical th}} \label{App:numerical th}

\begin{lemma} \label{lemma: fixed curve}
	If $(x,y)$ is an $(f,g)$-fixed point with $y>0$, then $x \leq q(y)$.
\end{lemma}

\begin{proof}
	Let $\epsilon \in (0,1)$ and let $(x,y)$ be a solution to \eqref{eq:f,g fixed} with $y>0$. In view of \eqref{eq:f} and \eqref{eq:g}, dividing the first equation of \eqref{eq:f,g fixed} with the second one yields
	\begin{align}\label{eq:divide f,g fixed}
	\frac{x}{y} = \frac{\lambda^{(1)}\left(1- \rho^{(1)}\left(1-x \right)\right) \cdot \lambda^{(2)}\left(1- \rho^{(2)}\left(1-y\right)\right)}{\lambda^{(2)}\left(1- \rho^{(2)}\left(1-y \right)\right) \cdot \lambda^{(1)}\left(1- \rho^{(1)}\left(1-x\right)\right)}
	\end{align} 
	which after some rearrangements implies 
	\begin{align}\label{eq:q_L(x)=q_J(y)}
	q_2(y)=q_1(x),
	\end{align}\
	where $q_2(\cdot)$ and $q_1(\cdot)$ are defined in \eqref{eq:q_L q_J}.
	In view of \eqref{eq:q_L q_J}, since $(x,y)$ is an $(f,g)$-fixed point, then 
	\begin{align}\label{eq:q_J leq 1}
	q_2(y)&=y\cdot \frac{\lambda^{(2)}\left( 1- \rho^{(2)}\left(1-y \right)\right)}{\lambda^{(2)}\left( 1- \rho^{(2)}\left(1-y \right)\right)} \notag \\
	&=g(\epsilon,x,y)\cdot \frac{\lambda^{(2)}\left( 1- \rho^{(2)}\left(1-y \right)\right)}{\lambda^{(2)}\left( 1- \rho^{(2)}\left(1-y \right)\right)} \notag \\
	&=\overbrace{\epsilon}^{\leq 1}\cdot \overbrace{\lambda^{(2)}\left( 1- \rho^{(2)}\left(1-y \right)\right)}^{\leq 1} \cdot \overbrace{\lambda^{(1)}\left( 1- \rho^{(1)}\left(1-x \right)\right)}^{\leq 1}\notag \\
	&\leq 1,
	\end{align}
	which together with Definition~\ref{def:q} and \eqref{eq:q_L(x)=q_J(y)} completes the proof.
	
\end{proof}

Let
\begin{align} \label{eq:converse 01}
\epsilon >  \inf\limits_{\substack{y \in (0,1]  \\ q_2(y)\leq 1}} \frac{y}{g(1,q(y),y)}.
\end{align}
There exists $y_0 \in (0,1]$ such that $q_2(y_0) \leq 1$ and
\begin{align}\label{eq:y_0}
y_0=\epsilon\cdot g(1,q(y_0),y_0)= \epsilon \cdot \lambda^{(2)}(1-\rho^{(2)}(1-y_0))\cdot\lambda^{(1)}(1-\rho^{(1)}(1-q(y_0))).
\end{align}
In view of \eqref{eq:q def}, 
\begin{align} 
q_1(q(y_0))=q_2(y_0),
\end{align}
which combined with \eqref{eq:f} and \eqref{eq:q_L q_J} yields
\begin{align} \label{eq:using q}
q(y_0) &= \frac{\lambda^{(1)}(1-\rho^{(1)}(1-q(y_0)))}{\lambda^{(1)}(1-\rho^{(1)}(1-q(y_0)))} \cdot y_0 \cdot \frac{\lambda^{(2)}(1-\rho^{(2)}(1-y_0))}{\lambda^{(2)}(1-\rho^{(2)}(1-y_0))} \notag \\
&=\epsilon \cdot \lambda^{(1)}(1-\rho^{(1)}(1-q(y_0)))\cdot \lambda^{(2)}(1-\rho^{(2)}(1-y_0)) \notag \\
&= f(\epsilon,q(y_0),y_0).
\end{align}
Thus, $(q(y_0),y_0)$ is a non-zero $(f,g)$-fixed point, which in view of Theorem~\ref{theorem:fix pt char} implies that $\epsilon > \epsilon^*_2$. Hence,  
\begin{align} \label{eq:th 1st part}
\epsilon^*_2 \leq \inf\limits_{\substack{y \in (0,1]  \\ q_2(y)\leq 1}} \frac{y}{g(1,q(y),y)}.
\end{align}

Next, let 
\begin{align} \label{eq:direct}
\epsilon <  \inf\limits_{\substack{y \in (0,1]  \\ q_2(y)\leq 1}} \frac{y}{g(1,q(y),y)}
\end{align}
and let $(x,y)$ be a solution to \eqref{eq:f,g fixed}. In what follows, we prove that $y=0$. Assume to the contrary that $y>0$. From Lemma~\ref{lemma: fixed curve} it follows that $x \leq q(y)$, which in view Lemma~\ref{lemma:monotonicity}, \eqref{eq:q_J leq 1} and \eqref{eq:direct} implies 
 \begin{align}\label{eq:numerical proof1}
y=g(\epsilon,x,y)\leq g(\epsilon,q(y),y)  < y,
\end{align}
in contradiction; thus, $y=0$. Next, consider two cases:
\begin{enumerate}
\item If $P_0=0$ or $\lambda^{(2)}(0) >0$, then Item~\ref{item: fix pt lemma0} of Lemma~\ref{lemma:fix pt}  implies that $x=0$. Hence, every $(f,g)$-fixed point satisfies $y=x=0$. In view of Theorem~\ref{theorem:fix pt char}, it follows that if \eqref{eq:direct} holds, then $\epsilon < \epsilon^*_2$, so
\begin{align}
\epsilon^*_2 \geq  \inf\limits_{\substack{y \in (0,1]  \\ q_2(y)\leq 1}} \frac{y}{g(1,q(y),y)}
\end{align} 
which with \eqref{eq:th 1st part} completes the proof when $P_0=0$ or $\lambda^{(2)}(0) >0$. 

\item If $P_0>0$ and $\lambda^{(2)}(0) =0$, it is not true in general that for every fixed point $(x,y)$, $y=0$ implies $x=0$. However, if in addition to \eqref{eq:direct},
\begin{align} \label{eq:direct2}
\epsilon <  \frac1{P_0}\cdot \inf_{(0,1]}\frac{x}{\lambda^{(1)}\left(1- \rho^{(1)}\left(1-x \right)\right)},
\end{align}
and $y=0$ for some fixed point $(x,y)$, then $x=0$. To see this, assume to the contrary that $x>0$. In view of \eqref{eq:f} and \eqref{eq:direct2} it follows that
\begin{align} \label{eq:x fixed}
x= f(\epsilon,x,0) = \epsilon \cdot P_0 \cdot \lambda^{(1)}\left(1- \rho^{(1)}\left(1-x \right)\right) < x
\end{align}
in contradiction; hence, if \eqref{eq:direct} and \eqref{eq:direct2} hold, $x=0$ thus $\epsilon < \epsilon^*_2$. This means that 
\begin{align}\label{eq:direct proof}
\epsilon^*_2 \geq \min\left\{ \inf\limits_{\substack{y \in (0,1]  \\ q_2(y)\leq 1}} \tfrac{y}{g(1,q(y),y)},\;\;\tfrac1{P_0}\cdot \inf_{(0,1]}\tfrac{x}{\lambda^{(1)}\left(1-  \rho^{(1)}\left(1-x \right)\right)}\right\}.
\end{align}
To complete the proof, we must show that when $P_0>0$ and $\lambda^{(2)}(0) =0$, then
\begin{align}\label{eq:converse proof}
\epsilon^*_2 \leq \min\left\{ \inf\limits_{\substack{y \in (0,1]  \\ q_2(y)\leq 1}} \tfrac{y}{g(1,q(y),y)},\;\;\tfrac1{P_0}\cdot \inf_{(0,1]}\tfrac{x}{\lambda^{(1)}\left(1-  \rho^{(1)}\left(1-x \right)\right)}\right\}.
\end{align}
If
\begin{align*}
\inf\limits_{\substack{y \in (0,1]  \\ q_2(y)\leq 1}} \frac{y}{g(1,q(y),y)} \leq \frac1{P_0}\inf_{(0,1]}\frac{y}{\lambda^{(1)}\left(1-  \rho^{(1)}\left(1-y \right)\right)},
\end{align*}
then \eqref{eq:converse proof} follows immediately from \eqref{eq:th 1st part}; hence we can assume that 
\begin{align} \label{eq:converse1}
\inf\limits_{\substack{y \in (0,1]  \\ q_2(y)\leq 1}} \frac{y}{g(1,q(y),y)} > \frac1{P_0}\inf_{(0,1]}\frac{y}{\lambda^{(1)}\left(1-  \rho^{(1)}\left(1-y \right)\right)}.
\end{align}
Let 
\begin{align} \label{eq:converse12}
\inf\limits_{\substack{y \in (0,1]  \\ q_2(y)\leq 1}} \frac{y}{g(1,q(y),y)} >\epsilon> \frac1{P_0}\inf_{(0,1]}\frac{y}{\lambda^{(1)}\left(1-  \rho^{(1)}\left(1-y \right)\right)},
\end{align}
and let $x_0 \in (0,1]$, such that $x_0=\epsilon \cdot P_0 \cdot \lambda^{(1)}\left(1-  \rho^{(1)}\left(1-x_0 \right)\right)$. Since $\lambda^{(2)}(0)=0$, it follows that $(x_0,0)$ is a fixed point with $x_0>0$, thus $\epsilon > \epsilon^*_2$. Since this is true for every $\epsilon> \tfrac1{P_0}\inf_{(0,1]}\frac{y}{\lambda^{(1)}\left(1-  \rho^{(1)}\left(1-y \right)\right)}$, then $\epsilon^*_2 \leq \tfrac1{P_0}\inf_{(0,1]}\frac{y}{\lambda^{(1)}\left(1-  \rho^{(1)}\left(1-y \right)\right)}$. In view of \eqref{eq:converse1}, it follows that \eqref{eq:converse proof} holds. This completes the proof for the $P_0>0$ and $\lambda^{(2)}(0) =0$ case.
\end{enumerate}

\section{Proof of Lemma~\ref{lemma:Multigap2cap}} \label{App:Multigap2cap}
The design rate of the $ L $-layer ensemble is given by
\begin{align*}
R=1-\sum_{i=1}^L\frac{\int_0^1\rho^{(i)}(x)\mathrm{d}x}{\int_0^1\lambda^{(i)}(x)\mathrm{d}x}\left (1-P_0^{(i)}\right ),
\end{align*}
and the threshold is $ \epsilon_L $. Thus, the overall gap to capacity is given by
\begin{align}\label{eq:Multigap2cap}
\begin{split}
\delta 
&= 1-R-\epsilon_L\\
&=\sum_{i=1}^L \frac{\int_0^1\rho^{(i)}(x)\mathrm{d}x}{\int_0^1\lambda^{(i)}(x)\mathrm{d}x}\left (1-P_0^{(i)}\right )-\epsilon_L\\
&=\epsilon_1+\delta_1+\sum_{i=2}^L\left (\epsilon_ia_s^{(i-1)}+\delta_i\right )\left (1-P_0^{(i)}\right ) -\epsilon_L\\
&=\delta_1+\sum_{i=2}^L \delta_i\left (1-P_0^{(i)}\right )+\epsilon_1+\sum_{i=2}^L \epsilon_ia_s^{(i-1)}\left (1-P_0^{(i)}\right ) -\epsilon_L.
\end{split}
\end{align}
In view of Construction~\ref{Const:multi}, $ P_0^{(i)}=\epsilon_{i-1}/\epsilon_{i} $ and $ a_s^{(i)}\leq 1 $. Hence,
\begin{align*}
\epsilon_1+\sum_{i=2}^L\left (\epsilon_ia_s^{(i-1)}\right )\left (1-P_0^{(i)}\right ) -\epsilon_L
&=\epsilon_1 +\sum_{i=2}^L a_s^{(i-1)}\left (\epsilon_i-\epsilon_{i-1}\right ) -\epsilon_L\\
&=\epsilon_1+\sum_{i=2}^L a_s^{(i-1)}\left (\epsilon_i-\epsilon_{i-1}\right )-\epsilon_L\\
&\hspace*{4mm} -\epsilon_2 -\sum_{i=3}^{L-1}\left (\epsilon_i-\epsilon_{i-1}\right ) +\epsilon_{L-1}\\
&=\sum_{i=2}^L \left (a_s^{(i-1)}-1\right )\left (\epsilon_i-\epsilon_{i-1}\right )\\
&\leq 0,
\end{align*}
which combined with \eqref{eq:Multigap2cap} completes the proof.

\section{Proof of Lemma~\ref{lemma:valid opt schedule}} \label{App:valid opt schedule}
To prove Lemma~\ref{lemma:valid opt schedule} we need the following lemma.

\begin{lemma}\label{lemma:valid schedule}
A scheduling scheme is valid if and only if, $\epsilon^{(1)}_{\mathrm{eff}}(y_l)<\epsilon_1$, for some iteration $l$.
\end{lemma}

\begin{proof}
Recall the definition of the type-1 threshold,
\begin{align} \label{eq:loc th def}
\epsilon_1 = \sup \{ \epsilon \colon x=\epsilon\lambda^{(1)}(1-\rho^{(1)}(1-x)) \text{ has no solution in } (0,1]\},
\end{align}
and let $x=\lim_{l \to \infty}x_l$ and $y=\lim_{l \to \infty}y_l$.
Since under every scheduling scheme $y_l$ is monotonically non-increasing in $l$, then in view of \eqref{eq:f with eff},
\begin{align*}
\exists l \in \mathbb{N}, \epsilon^{(1)}_{\mathrm{eff}}(y_l)<\epsilon_1 
&\Leftrightarrow \epsilon^{(1)}_{\mathrm{eff}}(y) < \epsilon_1 \\
&\Leftrightarrow x=\epsilon^{(1)}_{\mathrm{eff}}(y) \lambda^{(1)}(1-\rho^{(1)}(1-x)) \text{ has no solution for }x \in (0,1] \\
&\Leftrightarrow \lim\limits_{l \to \infty}x_l=0.
\end{align*} 
\end{proof}
We proceed with the proof of Lemma~\ref{lemma:valid opt schedule}.
Let $ (x_l,y_l) $ be defined as in \eqref{eq:opt schedule}, let $ (x,y)=\lim_{l\to\infty}(x_l,y_l) $, and assume in contradiction that $\epsilon^{(1)}_{\mathrm{eff}}(y) \geq \epsilon_1$. Since $\eta>0$, letting $l \to \infty$ in \eqref{eq:opt schedule} implies that $(x,y)$ is a non-trivial $(f,g)$-fixed point. However, in view of Theorem~\ref{theorem:fix pt char}, if $\epsilon<\epsilon_2$, then every $(f,g)$-fixed point is the trivial point, in contradiction. Thus,  $\epsilon^{(1)}_{\mathrm{eff}}(y) < \epsilon_1$ which, due to Lemma~\ref{lemma:valid schedule}, completes the proof.

\section{Proof of Lemma~\ref{lemma:opt schedule}} \label{App:opt schedule}
Let $\left\{l^{(1)}_k\right\}_{k=1}^{N^{(1)}_{L_2}}$ and  $\left\{l^{(2)}_k\right\}_{k=1}^{N^{(2)}_{L_2}}$ be the type-2-update iterations of the scheduling scheme described in \eqref{eq:x_l_k} and in some arbitrary valid scheduling scheme, receptively. We need to show that $N_{2}^{(1)} \leq N_{2}^{(2)}$. 
To proceed we need the following lemmas:
\begin{lemma}\label{lemma:x_s monotonic}
	$ x_s(\epsilon) $ as defined in Definition~\ref{def:h_eps} is monotonic non-decreasing in $ \epsilon\;. $
\end{lemma}
\begin{proof}
	Let $ \epsilon_1\leq \epsilon_2 $, and consider $ x_s(\epsilon_1),\;x_s(\epsilon_2)\;. $ In view of Definition~\ref{def:h_eps}, 
	\begin{align*}
		h_{\epsilon_2}\left (x_s(\epsilon_1)\right ) 
		&\triangleq \epsilon_2\lambda^{(1)}(1-\rho^{(1)}(1-x_s(\epsilon_1)))-x_s(\epsilon_1)\\
		&\geq		\epsilon_1\lambda^{(1)}(1-\rho^{(1)}(1-x_s(\epsilon_1)))-x_s(\epsilon_1)\\
		&\triangleq	h_{\epsilon_1}\left (x_s(\epsilon_1)\right )\\
		&\geq 		0\;. 
	\end{align*}
	Thus, $ x_s(\epsilon_2)\triangleq \max\{x\in[0,1]\;\colon h_{\epsilon_2}\left (x\right )\geq 0\}\geq x_s(\epsilon_1)\;. $
\end{proof}
\begin{lemma} \label{lemma:opt schdule induction}
Let
\begin{align}\label{eq:eps_k 1 and 2}
\begin{split}
&\varepsilon^{(1)}_k=\epsilon^{(1)}_{\mathrm{eff}}\left(y_{l^{(1)}_k}\right),\quad 1\leq k \leq N^{(1)}_{2}, \\
&\varepsilon^{(2)}_k=\epsilon^{(1)}_{\mathrm{eff}}\left(y_{l^{(2)}_k}\right),\quad 1\leq k \leq N^{(2)}_{2}.
\end{split}
\end{align}
Then, for every $ 1 \leq k \leq \min\left(N^{(1)}_{2},N^{(2)}_{2} \right) $,
\begin{align} \label{eq:opt schdule induction}
y_{l^{(1)}_k} \leq y_{l^{(2)}_k},\quad\text{and}\quad \varepsilon^{(1)}_k\leq \varepsilon^{(2)}_k,\quad\text{and}\quad x_{l^{(1)}_k} \leq x_{l^{(2)}_k}\;.
\end{align}
\end{lemma}

\begin{proof}
By induction on $1 \leq k \leq \min\left(N^{(1)}_{2},N^{(2)}_{2} \right)$. 
In the first type-2 update, we have $y_{l^{(1)}_1}=1=y_{l^{(2)}_1}$ and $ \varepsilon_1\triangleq \epsilon^{(1)}_{\mathrm{eff}}\left (y_{l^{(1)}_1}\right )=\epsilon$. Thus, in view of Definition~\ref{def:h_eps}, in the first type-2 update $ =x_{l^{(2)}_1}\geq x_s(\epsilon)=x_{l^{(1)}_1}$.
Hence \eqref{eq:opt schdule induction} holds for $k=1$. Assume correctness for some type-2 update $k<\min\left(N^{(1)}_{2},N^{(2)}_{2} \right)$, and consider update $k+1$. In view of Lemma~\ref{lemma:monotonicity}, \eqref{eq:y_l_k}-\eqref{eq:x_l_k}, and the induction assumption,
$y_{l^{(1)}_{k+1}}= g\left(\epsilon,x_{l^{(1)}_{k}},y_{l^{(1)}_{k}}\right) \leq g\left(\epsilon,x_{l^{(2)}_{k}},y_{l^{(2)}_{k}}\right) =y_{l^{(2)}_{k+1}}\;,$
which together with \eqref{eq:eps_k} implies that $ \varepsilon^{(1)}_{k+1}\triangleq \epsilon^{(1)}_{\mathrm{eff}}\left (y_{l^{(1)}_{k+1}}\right )\leq \epsilon^{(1)}_{\mathrm{eff}}\left (y_{l^{(2)}_{k+1}}\right )\triangleq\varepsilon^{(2)}_{k+1}.$ In view of Lemma~\ref{lemma:x_s monotonic}, it follows that $ x_{l^{(1)}_{k+1}}\triangleq x_s(\varepsilon^{(1)}_{k+1})\leq x_s(\varepsilon^{(2)}_{k+1})\leq x_{l^{(2)}_{k+1}}.$ By induction, we complete the proof.
\end{proof}
We proceed with the proof of Lemma~\ref{lemma:opt schedule}.
Assume, on the contrary, that  $N_{2}^{(1)} > N_{2}^{(2)}$. Lemma~\ref{lemma:opt schdule induction} and the monotonicity of $ \varepsilon_k $ in $ k $ imply that
\begin{align} \label{eq:eps^1 > eps_L}
\epsilon^{(2)}_{N_{2}^{(2)}} \geq \epsilon^{(1)}_{N_{2}^{(2)}}  \geq \epsilon^{(1)}_{N_{2}^{(1)}-1} \geq \epsilon_1,
\end{align}
which, in view of Lemma~\ref{lemma:valid schedule} yields that the scheduling scheme indexed by $\left\{l^{(2)}_k\right\}_{k=1}^{N^{(2)}_{2}}$ is not valid, in contradiction. Thus, $N_{2}^{(1)} \leq N_{2}^{(2)}$.

\section{Proof of Lemma~\ref{lemma:x_s c.a.}} \label{App:x_s c.a.}
In view of Definition~\ref{def:h_eps}, let $h^{(k)}_\epsilon(x)=\epsilon\lambda^{(k)}(1-\rho^{(k)}(1-x))-x$. Eq. \eqref{eq:c.a. known} yields
\begin{align}\label{eq:c.a. h_eps}
h_\epsilon(x)=\lim_{k \to \infty}h^{(k)}_\epsilon(x)= \left\{ 
\begin{array}{ll}
\left( \frac{\epsilon}{\epsilon_1}-1\right) x & 0 \leq x \leq \epsilon_1 \\
\epsilon- x & \epsilon_1 \leq x \leq \epsilon 
\end{array}
\right.
\end{align}
For every $k \in \mathbb{N}$, and $x\in (\epsilon,1]$, 
\begin{align*}
h_\epsilon^{(k)}(x) 
&= \epsilon\lambda^{(k)}(1-\rho^{(k)}(1-x))-x \\ \notag
&\leq \epsilon-x \\ \notag
&<0,
\end{align*}
Thus
\begin{align}\label{eq:x_s^k <eps}
x_s^{(k)}(\epsilon) \leq \epsilon ,\quad \forall k \in \mathbb{N}.
\end{align}
In addition, for every $0<a<\epsilon$ there exists $K_0$ such that 
\begin{align*}
h_\epsilon^{(k)}(\epsilon-a)>0 ,\quad \forall k \geq K_0,
\end{align*}
so $x_s^{(k)}(\epsilon) \geq \epsilon-a$, for every $ k \geq K_0$; hence,
\begin{align} \label{eq:x_s liminf}
\liminf_{k \to \infty} x_s^{(k)}(\epsilon) \geq \epsilon.
\end{align}
Combining \eqref{eq:x_s^k <eps} and \eqref{eq:x_s liminf} implies that $\lim_{k \to \infty}x^{(k)}_s(\epsilon)$ exists, and completes the proof.

\end{document}